%% file: ms.tex
\documentclass{siamonline1116}

\input{ex_shared}

\ifpdf
\hypersetup{
  pdftitle={\TheTitle},
  pdfauthor={\TheAuthors}
}
\fi

\begin{document}

\maketitle

\begin{abstract}
The mutual intensity and its equivalent phase-space representations
quantify an optical field's state of coherence and are important tools
in the study of light propagation and dynamics, but they can only be
estimated indirectly from measurements through a process called
coherence retrieval, otherwise known as phase-space tomography.  As
practical considerations often rule out the availability of a complete
set of measurements, coherence retrieval is usually a challenging
high-dimensional ill-posed inverse problem. In this paper, we propose
a trace-regularized optimization model for coherence retrieval and a
provably-convergent adaptive accelerated proximal gradient algorithm
for solving the resulting problem.  Applying our model and algorithm
to both simulated and experimental data, we demonstrate an improvement
in reconstruction quality over previous models as well as an increase
in convergence speed compared to existing first-order methods.
\end{abstract}



\section{Introduction}
\emph{Coherence retrieval} is the unified mathematical treatment for
two analogous physical measurement processes: estimating the density
matrix of a quantum state, and reconstructing the mutual intensity of
partially coherent light.  Many, if not most coherence retrieval
methods in optics originate from the seminal phase-space tomography
approach~\cite{Raymer1994, McAlister1995}.  This reconstructed mutual
intensity, or equivalent phase-space representations such as the
Wigner distribution, is highly useful in that it can predict the
three-dimensional distribution of light intensity after propagation
through \emph{any} known linear optical system~\cite{Born2005,
  Goodman1985}.  Likewise, knowledge of both the input and output
coherence state of a system enables study of its
dynamics~\cite{Clark2014}.  Furthermore, the mutual intensity's
phase-space representations enable intuitive understanding of light
propagation~\cite{Testorf2009} and form the physical basis for light
field imaging in the field of computational
photography~\cite{Ng2005_2, Ng2005, Levoy2009, Zhang2009}.

In a coherence retrieval experiment, stationary quasimonochromatic
light, \emph{i.e.,} narrow in temporal frequency with mean wavelength
$\lambda$, in some unknown state of coherence is sent through one or
more known optical systems, with the intensity measured at their
outputs.  The spatial coherence properties of this light source,
\emph{e.g.,} the mutual intensity, are then recovered from these
measurements~\cite{Thompson1957, Asakura1972, Bartelt1980,
  Michalski1985, Li1988, Barreiro1993, Raymer1994, McAlister1995,
  Dragoman2003, Rydberg2007, Tian2012, Waller2012, Tian2013,
  Stoklasa2014}.  More concretely, the electric field for
quasimonochromatic light on some source plane is described by a
complex-valued stochastic function $U\parens{\vec p}$, where $\vec
p\in\realvectors{2}$ denotes a position on this plane.  For the
special case of a fully coherent field, such as light from a laser,
$U\parens{\vec p}$ is deterministic, yielding the amplitude and phase
for the wave function.  However, in this paper we consider the general
case where the field is not necessarily coherent, such as one
emanating from a light emitting diode (LED), and thus $U\parens{\vec
  p}$ can take on many possible values~\cite{Mandel1995}.\footnote{The
  analogous case in quantum tomography of a partially coherent field
  is a mixed state, whereas the fully coherent field corresponds to a
  pure state.}  The mutual intensity is the correlation function
describing this electric field, given by:
\begin{equation}
J\parens{\vec p_1,\vec p_2} = \expectation\brackets[\big]{U\parens{\vec
    p_1}U^\conj\parens{\vec p_2}}=\iint
\upsilon_1\upsilon_2^\conj 
P_U\parens{\upsilon_1,\upsilon_2;\vec
  p_1,\vec p_2}\dd{\upsilon_1}\dd{\upsilon_2}\text{,}
\label{eq:mutualintensity}
\end{equation}
where $\expectation\brackets{\cdot}$ denotes the expected value
(\emph{i.e.,} the ensemble average), $^\conj$ denotes the complex
conjugate, and $P_U\parens{\upsilon_1,\upsilon_2;\vec p_1,\vec p_2}$
is the joint probability density function for $U\parens{\vec
  p_1}=\upsilon_1$ and $U\parens{\vec p_2}=\upsilon_2$.  Other
phase-space representations such as the Wigner distribution and
ambiguity function can be obtained from the mutual intensity via
Fourier transforms~\cite{Testorf2009}.

In typical experimental conditions, the optical field propagates
linearly from a source point $\vec p$ to one of the $M$ measurement
points $\vec q_m\in\realvectors{3}$, traversing various parts of an
experimental apparatus.  This linear relationship is characterized by
the possibly stochastic kernel $K\parens{\vec p,\vec q_m}$, known as
the transmission function~\cite{Born2005} or the amplitude spread
function~\cite{Goodman1985}:
\begin{equation*}
V\parens{\vec q_m} = \int U\parens{\vec
  p}K\parens{\vec p,\vec q_m}\dd{\vec p}
\end{equation*}
where $V\parens{\vec q_m}$ describes the output field at position
$\vec q_m$.  The resulting intensity is the expected value of the
magnitude of the field squared, yielding the Hopkins integral:
\begin{equation*}
I\parens{\vec q_m} =
\expectation\brackets[\big]{V\parens{\vec
    q_m}V^\conj\parens{\vec q_m}} = \iint
\expectation\brackets[\big]{U\parens{\vec
    p_1}U^\conj\parens{\vec
    p_2}}\expectation\brackets[\big]{K\parens{\vec p_1,\vec q_m}K^\conj\parens{\vec
  p_2,\vec q_m}} \dd{\vec p_1}\dd{\vec p_2}
\end{equation*}
if we assume that the statistical fluctuations in the optical field
are independent of the statistical fluctuations in $K$, typically the
case in practical scenarios.  

To enable numerical computation, we choose some basis $\set{\xi_i}_i$
and approximate the field on the source plane as $U\parens{\vec
  p}\approx\sum_{i=1}^N\nolimits u_i\xi_i\parens{\vec p}$ where
$\set{u_i}_i$ are complex-valued random variables.  This naturally
yields a discretization for the mutual
intensity~\eqref{eq:mutualintensity} as a Hermitian positive
semidefinite matrix $\mtx X = \expectation\parens*{\vec u\vech u\,}$
such that
\begin{equation}
J\parens{\vec p_1,\vec p_2} \approx \vect\xi\parens{\vec p_1}\mtx
X\vec\xi^*\parens{\vec p_2}\text{,}
\label{eq:discretemutualintensity}
\end{equation}
where $^\T$ denotes transpose, $^\HT$ denotes conjugate transpose,
$\vec u=\brackets[\big]{u_1,u_2,\ldots,u_N}^\T\in\complexvectors{N}$,
and $\vec\xi\parens{\vec p}=\brackets[\big]{\xi_1\parens{\vec
    p},\xi_2\parens{\vec p},\ldots,\xi_N\parens{\vec p}}^\T$.  The
basis $\{\xi_i\}_i$ is chosen based on both ease of computation and
efficiency in representing the field and mutual intensity; typical
basis functions can be constructed from $\sinc$ functions, Hermite
functions and the prolate spheroidal functions.

Using the definition of the mutual
intensity~\eqref{eq:mutualintensity} as well as our discretization
scheme~\eqref{eq:discretemutualintensity}, we obtain:
\begin{equation}
I\parens{\vec q_m} \approx \iint \vect\xi\parens{\vec p_1}\mtx
X\vec\xi^*\parens{\vec p_2}\expectation\brackets[\big]{K\parens{\vec
    p_1,\vec q_m}K^\conj\parens{\vec p_2,\vec q_m}} \dd{\vec
  p_1}\dd{\vec p_2}
\label{eq:halfway}
\end{equation}
Now let us define a vector $\vec k_m\in\mathbb{C}^N$ with $n$th
element $\vec k_m[n]=\int \xi_n\parens{\vec p}K\parens{\vec p,\vec
  q_m} \dd{\vec p}$ being a random variable with probability density
$P_{\vec k_m}\parens*{\vec{\hat k}_m}$ for each possible realization
$\vec{\hat k}_m\in\complexvectors{N}$.  We can then construct a discretized
measurement operator that characterizes propagation of light from the
source to point $\vec r_m$:
\begin{equation*}
\mtx K_m = \expectation\brackets[\Big]{\vec k_m^\conj\vect k_m} = \int P_{\vec
  k_m}\parens*{\vec{\hat k}_m}\vec {\hat k}_m^\conj\vect {\hat
  k}_m\dd{\vec{\hat k}_m}\text{.}
\end{equation*}
Substituting these new definitions back into \eqref{eq:halfway}, we
obtain that the intensity is an inner product of the source mutual
intensity matrix $\mtx X$ and the discretized measurement operator
$\mtx K_m$ (\emph{i.e.,} the sum of the elementwise product of $\mtx
X$ with the complex conjugate of $\mtx K_m$):
\begin{equation*}
I\parens{\vec q_m} \approx \innerproduct{\mtx K_m}{\mtx X} =
\trace\parens*{\mtxh K_m\mtx X}
\end{equation*}
where $\innerproduct{\mtx A}{\mtx B}$ denotes the inner product of
matrices $\mtx A$ and $\mtx B$, and $\trace\parens*{\cdot}$ denotes
the trace. We note that the $\mtx K_m$s are Hermitian, but we keep the
conjugate transpose notation here to be consistent with later uses of
inner products in the paper, which involve matrices which may or may
not be Hermitian.  In many situations, the $\mtx K_m$s are low-rank;
for example, rank-one operators were used
in~\cite{Tian2012,Zhang2013}.

We can model measurement noise as an additive term:
\begin{equation}
y_m = \trace\parens{\mtxh K_m\mtx X}+n_m,
\label{eq:measurement}
\end{equation}
where $y_m$ denotes the $m$th measured value, and
  $n_m$ denotes the noise term.  The $n_m$s can be well-approximated
by zero-mean normal distributions with standard deviations $\sigma_m$s
if measurements are from a standard camera sensor pixel with at least
$\sim10$ photons recorded and the noise level before quantization is
much larger than a single quantization level, \emph{i.e.,} if the
noise is predominantly Poisson and the rate parameter is high enough.

The problem of coherence retrieval is then about recovering $\mtx X$
from the measured intensities $\{y_m\}_m$, which can be formulated as
solving a weighted least-squares semidefinite problem:
\begin{equation}
\label{eq:0}
\begin{aligned}
&\minimize[\mtx
    X]&&\sum_{m=1}^M\frac{1}{2\sigma_m^2}\parens*{\trace\parens{\mtxh
      K_m\mtx X}-y_m}^2&&&&\subjectto&&\mtx X\in\PSD,
\end{aligned}
\end{equation}
where $\SN$ is the set of $N\times N$ Hermitian matrices, $\PSD$ is
the set of $N\times N$ positive semidefinite matrices in $\SN$, and
$M$ is the total number of measurements.  

Unlike the phase retrieval problem~\cite{GerchbergSaxton1972,
  Fienup1982, Chai2011, Candes2013}, the rank of $\mtx X$ is generally
bigger than one; only in the special case of a coherent field do we
have that $\rank\mtx X=1$, and that is what is referred to as phase
retrieval.  With coherent light in phase retrieval, the goal is to
reconstruct the field, represented by a \emph{deterministic} vector
$\vec u$; the lifting approach~\cite{Chai2011, Candes2013}
reformulates the problem as seeking the rank-one matrix $\vec u\vech
u$ instead, with low-rank promoters to rule out higher rank solutions.
However, with coherence retrieval, we consider the more general
partially coherent fields~\cite{Mandel1995}---$\vec u$ is a
\emph{stochastic} vector, and we seek to reconstruct $\mtx
X=\expectation\parens*{\vec u\vech u}$, which can have rank greater
than $1$ due to $\mtx X$ being an expectation across many possible
rank-one matrices.  Hence, our goal in coherence retreival is not to
recover a single $N$-dimensional vector lifted to matrix form as was
the case in phase retrieval, but rather to recover an $N\times N$
matrix, which is not necessarily rank-one and might not even be low
rank.

In practice, coherence retrieval is usually an ill-posed inverse
problem. Recovering the $N\times N$ coefficient matrix $\mtx X$
requires $O\parens{N^2}$ measurements and thus a forward operator of
size $O\parens{N^4}$, which would require prohibitively high
computational and storage costs.  Furthermore, even with smaller
values of $N$, sometimes it is not straight-forward to take enough
measurements for the forward operator to have a reasonable condition
number.  For example, in translation-only phase-space tomography, the
camera would need to be translated infinitely far away as well as
behind the source to capture a complete set of measurements.

To tackle these difficulties, one could introduce optical elements
such as lenses~\cite{McAlister1995, Marks2000}, but this can also
introduce systematic error in the solution without accurate enough
calibration of optical element positioning and aberrations.  Another
approach is to add a regularization term to \eqref{eq:0}; for example,
nuclear norm regularization is used in \cite{Tian2012} to promote
low-rank solutions. However, the low-rank prior may not be reasonable
in many coherence retrieval scenarios in practice, and a single scalar
parameter for the regularizer may not be flexible enough, either.

This paper aims to develop an effective approach for coherence
retrieval. We propose a trace-regularized model based on a penalty
term physically analogous to the total intensity of light after
passing through a chosen (virtual) linear system. The proposed
regularization is motivated by the concept that for a set of solutions
with similar likelihood, the one which is least energetic is likely to
be closer to the truth---the extra intensity could simply be an
artifact due to noise.  Flexibility in choosing an arbitrary virtual
system enables encoding additional \emph{a priori} information about
the solution as well.  These concepts lead to the following
trace-regularized optimization formulation for coherence retrieval:
\begin{equation}
\begin{aligned}
&\minimize[\mtx X] && 
\sum_{m=1}^M\frac{1}{2\sigma_m^2}\brackets*{\trace\parens[\big]{\mtxh K_m\mtx X}-y_m}^2+
\mu\trace\parens[\big]{\mtxh R\mtx X}
&&&& \subjectto && \mtx X\in\PSD,
\end{aligned}
\label{eq:regularized00}
\end{equation}
where $\mu>0$ is the penalty parameter and $\mtx R\in\PSD$ is a
virtual measurement operator encoding our choice of virtual system.
Given \emph{a priori} information, we can set $\mtx R$ to a value
other than $\mtx I$ to penalize unlikely values of $\mtx X$.
Candidates for the matrix $\mtx R$ include diagonal weighting matrices
as well as difference operators or high-pass filters constructed from
wavelet tight frames.  While previous
literature~\cite{Chai2011,Candes2013,Tian2012} have used the $\mtx
R=\mtx I$ special case to promote low-rank solutions, our goal here is
not to specifically promote low-rank solutions, but rather to promote
less noisy solutions and encode other \emph{a priori} information such
as smoothness or soft constraints on support.

Although many convex optimization methods can be applied to solve our
new model \eqref{eq:regularized00}, we present an efficient
first-order scheme tailored for this particular problem.  It is based
on the accelerated proximal gradient (APG) method~\cite{beck2009fast}
with adaptive restart~\cite{Odonoghue2015, su2014differential}, and we
introduce a new restart criterion so that under some mild conditions
on the measurement operators $\mtx K_m$s, we can prove that the
proposed algorithm is globally convergent---the generated sequence
converges to a global minimum of
\eqref{eq:regularized00}. Furthermore, we also propose a sufficient
condition for when our algorithm is provably \emph{linearly}
convergent. Numerical results show the advantage of the proposed model
with both simulated and experimental data, and the proposed numerical
algorithm also outperforms other state-of-the-art methods in terms of
computational efficiency for these data sets.

The rest of the paper is organized as follows: we explain the
principles guiding our trace-regularized approach in
\Cref{section:model}, give a numerical algorithm to solve the proposed
model in \Cref{section:algorithm}, analyze its convergence in
\Cref{section:convergence}, present numerical experiments for both
simulated and experimental data in \Cref{section:results}, and discuss
the results in \Cref{section:discussion}.

\section{The trace-regularized coherence retrieval model}
\label{section:model}

For notational brevity, we first summarize our proposed convex problem
for regularized coherence retrieval as follows:
\begin{equation}
\begin{aligned}
&\minimize[\mtx X] && 
\tfrac12\norm[\big]{\mathscr{A}\parens{\mtx X}-\vec b}_2^2 +
\mu\trace\parens[\big]{\mtxh R\mtx X}
&&&& \subjectto && \mtx X\in\PSD,
\end{aligned}
\label{eq:regularized}
\end{equation}
where linear operator $\mathscr{A}:\SN\rightarrow\realnumbers^M$ has
$m$th element $\A{\mtx X}[m]=\trace\parens[\big]{\mtxh K_m\mtx
  X}/\sigma_m$, and vector $\vec b\in\mathbb{R}^M$ has $m$th element
$\vec{b}[m]=y_m/\sigma_m$.  

Our choice of using regularizers of the form
$\trace\parens[\big]{\mtxh R\mtx X}$ derives from several motivations.
The first is that $\trace\parens[\big]{\mtxh R\mtx X}$ for positive
semidefinte $\mtx R$ corresponds to a physical quantity: the resulting
intensity if the source is channeled through an optical apparatus
defined by measurement operator $\mtx R$.  This interpretation allows
us to pose the problem as seeking the least physically energetic
solution that satisfies the measurements to an acceptable degree.
Minimizing the energy is a common approach in many inverse problems,
and this interpretation is more widely applicable than the
rank-minimization interpretation~\cite{Chai2011, Candes2013,
  Tian2012}---not many partially coherent fields are exactly low-rank
despite having decaying eigenvalues.

Furthermore, by not being restricted to setting $\mtx R=\mtx I$, we
enable some flexibility in encoding other \emph{a priori} information.
For example, we can encode the unequal likelihood of the spatial basis
functions by setting $\mtx R=\mtx W$, where $\mtx W$ is a diagonal
matrix whose entries give the relative negative log-likelihood of the
spatial basis functions.  This can be used to enforce a \emph{soft}
constraint on the spatial support of the solution (see for
example~\cite{LeBesnerais2008}) if the basis functions are spatially
localized, \emph{e.g.,} $\sinc$ functions.  Unlike a hard spatial
support constraint, this soft constraint allows us to embed
uncertainty into the solution---for example, imaging an aperture using
a lens will likely not result in a sharp aperture due to aberrations,
so forcing the solution to be zero outside the image of the aperture
is overly restrictive.

The well-studied Gaussian Schell-model source and their relatives (see
for example \cite{Starikov1982, Simon1993, Qiu2005, Sahin2012}) have
smooth underlying wave functions, and the statistics of natural images
also suggest a decay property in amplitude of the Fourier transform of
the intensity profile~\cite{Burton1987, Field1987, Tolhurst1992}.
These optical fields will be more likely to have lower energy content
at higher spatial frequencies, and this suggests setting $\mtx R$ in
such a way so that $\trace\parens[\big]{\mtxh R\mtx X}$ is more
sensitive to high frequency content.  We can do this by setting $\mtx
R=\mtx D=\mtxh H\mtx H$, where $\mtx H$ is defined such that the
continuous function $\vect\xi\parens{\vec p}\mtx H\vec u$ is equal to
$\vect\xi\parens{\vec p}\vec u$ with its high frequency components
boosted.  Therefore, $\trace\parens[\big]{\mtxh H\mtx H\mtx
  X}=\trace\parens[\big]{\mtx H\mtx X\mtxh H}$ is the trace of matrix
$\mtx X$ after boosting its high frequency components, resulting in a
higher penalty value for nonsmooth solutions. In this paper, we obtain
good results from the use of simple matrices for $\mtx H$ such as one
that extracts the high-pass component using the Haar wavelet, whereas
a more powerful and flexible approach would be to design $\mtx R$
based on the concept of high-pass filters of wavelet tight
frames~\cite{dong2010mra}, which has been shown to have a close
relationship to the difference operators~\cite{cai2012image}.

We now present a numerical example that shows how trace regularization
and choice of $\mtx R$ affects reconstruction quality in an idealized
coherence retrieval scenario wherein closed-form solutions exist, in
order to avoid complications due to algorithmic differences.  In this
example, we seek to reconstruct a one-dimensional Gaussian
Schell-model source~\cite{Starikov1982} with parameter $\beta=1$ from
simulated noisy measurements through an ideal apparatus, \emph{i.e,,}
the linear operator $\A{\cdot}$ has unit singular values and $\vec b$
is drawn from an i.i.d. Gaussian ensemble.  For simplicity, we chose a
spatial basis consisting of the first \num{32} Hermite functions
$\phi_n\parens{x}$ given in~\cite{Starikov1982}, and thus the ground
truth $\mtx X_\star\in\PSD$ is a diagonal $32\times32$ matrix whose
$n$th element is equal to $2^{n-1}\parens*{3+\sqrt{5}}^{1-n}$.  Note
that while $\mtx X_\star$ has decaying eigenvalues, it is not low
rank.  To see the effect of regularization, we consider the following
four closed-form solutions:
\begin{enumerate}
\item $\mtx X_{\text U}=\mtx X_\star+\sigma \mtx G$ is the solution to
  \eqref{eq:0} where positivity is ignored, \emph{i.e.,} where we
  replace $\PSD$ with $\SN$; $\sigma$ gives the noise level and $\mtx
  G$ is drawn from a Gaussian unitary ensemble.  We use this primarily
  as a point of reference for the other reconstruction approaches.
\item $\mtx X_{\text 0}=\Proj_{\PSD}\parens*{\mtx X_\star+\sigma\mtx
  G}$ is the solution to \eqref{eq:0}.
\item $\mtx X_{\text I}=\Proj_{\PSD}\parens*{\mtx X_\star+\sigma\mtx
  G-\mu\mtx I}$ is the solution to \eqref{eq:regularized} with $\mtx R=\mtx I$.
\item $\mtx X_{\text D}=\Proj_{\PSD}\parens*{\mtx X_\star+\sigma\mtx
  G-\mu\mtx D}$ is the solution to \eqref{eq:regularized} with $\mtx
  R=\mtx D=\mtxt {\hat D}\mtx {\hat D}$, where $\mtx
  {\hat D}$ acting on a discretized field $\vec u$ is equivalent to
  performing a derivative on the continuous quantity that $\vec u$
  represents.  This choice of $\mtx R$ is used to incorporate the idea
  that Gaussian Schell-modes are generally smooth, and hence contain
  less high frequency content.  With Hermite functions as the spatial
  basis, the $(i,j)$ entry of $\mtx{\hat D}$ is given by:
  \begin{equation*}
    \hat D_{i,j} = 
\begin{cases} 
-\sqrt{j+1}, & \text{if } i=j+1\\
\sqrt{j}, & \text{if } i=j-1\\
0, & \text{otherwise.}
\end{cases}
  \end{equation*}
\end{enumerate}
The noise level parameter $\sigma$ was allowed to take on \num{11}
different values exponentially equally spaced between
$10^{-3}\norm{\mtx X_\star}_\frobenius$ and $10^{-1}\norm{\mtx
  X_\star}_\frobenius$.  For each noise level, the regularization
parameter $\mu$ that minimized the average reconstruction error was
found using the bisection method, with different parameters for the
$\mtx R=\mtx I$ and $\mtx R=\mtx D$ cases.  The resulting spread of
reconstruction error across \num{256} realizations for each method and
noise level is shown in \cref{fig:regularizer:trends}.  While adding a
$\trace\parens{\mtx X}$ term to the optimization results in an
improvement in reconstruction quality over the unregularized result,
using $\trace\parens[\big]{\mtxh D\mtx X}$ results in even less error,
due to incorporating \emph{a priori} information about the smoothness
of the solution.  We also display a scatter of the reconstruction
error as a function of both $\trace\parens{\mtx X}$ and
$\trace\parens[\big]{\mtxh D\mtx X}$ in \cref{fig:scatter:trace}.
Note that the trace regularization terms and the reconstruction error
are positively correlated in the unregularized case; while desiring a
less energetic solution is a good physical rule of thumb, this
correlation could explain why it is good mathematically with further
study.

\begin{figure}[htbp]
\centering
\includegraphics{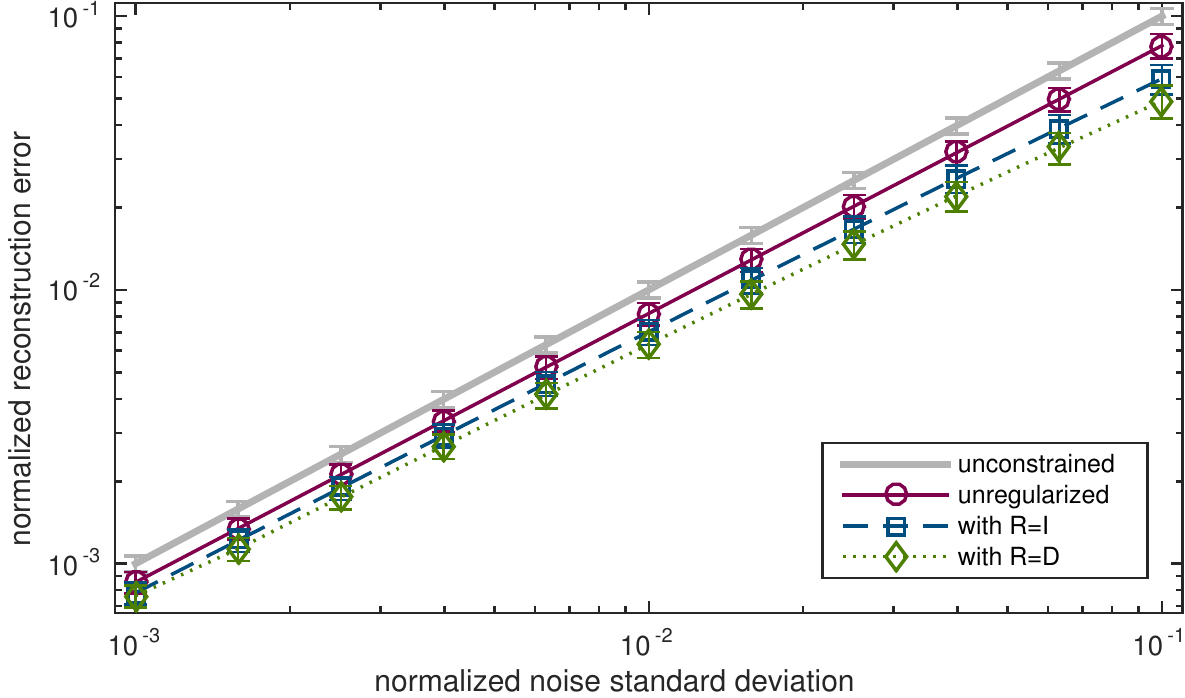}
\caption{Normalized reconstruction error $\norm{\mtx
    X-\mtx X_\star}_\frobenius/\norm{\mtx X_\star}$ as
  a function of the normalized noise level $\sigma/\norm{\mtx
    X_\star}$.  The mean is plotted for the four different solution
  methods, where $\mtx X$ is set to $\mtx X_{\text U}$,
  $\mtx X_{\text 0}$, $\mtx X_{\text I}$ and $\mtx X_{\text D}$ for the
  ``unconstrained'', ``unregularized'', ``with R=I'' and ``with R=D''
  cases, respectively.  Error bars indicate three standard deviations
  for each case.\label{fig:regularizer:trends}}
\end{figure}

\begin{figure}[htbp]
\centering
\includegraphics{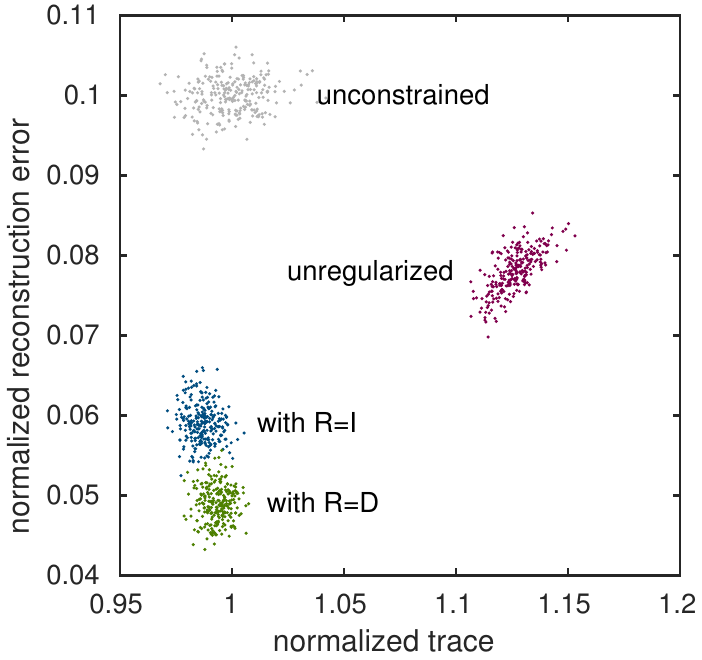}
\hfill
\includegraphics{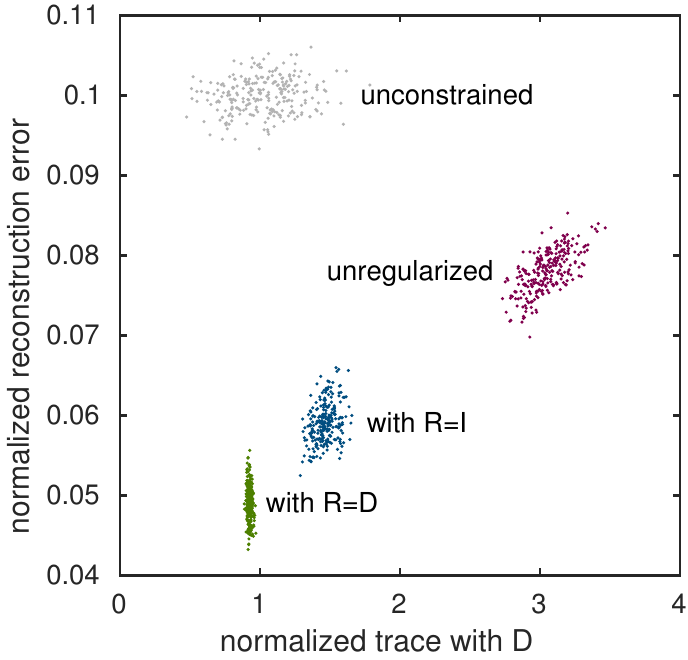}
\caption{Normalized reconstruction error $\norm{\mtx X-\mtx
    X_\star}_\frobenius/\norm{\mtx X_\star}$ versus normalized trace
  $\trace\parens{\mtx X}/\trace\parens{\mtx X_\star}$ (left) and
  $\trace\parens[\big]{\mtxh D\mtx X}/\trace\parens{\mtxh D\mtx
    X_\star}$ (right) for the four different reconstruction
  approaches, represented as a scatter plot showing each of the \num{256}
  realizations.\label{fig:scatter:trace}}
\end{figure}

\section{Numerical algorithm}

\label{section:algorithm}

We will now briefly review the accelerated proximal gradient (APG)
method \cite{beck2009fast} and then present a novel restart condition
for solving \eqref{eq:regularized}.

\subsection{The accelerated proximal gradient method}
Define 
\begin{equation}
f(\mtx X)  = \tfrac12\norm[\big]{\A{\mtx X}-\vec b}_2^2 + \mu\innerproduct{\mtx R}{\mtx X}, \quad
g(\mtx X)  = \iota_{\PSD}(\mtx X) \label{eq:fandg}
\end{equation}
where $\innerproduct{\mtx A}{\mtx B}=\trace\parens{\mtxh A\mtx B}$ is
the inner product and $\iota_{\PSD}\parens{\cdot}$ is the indicator
function for the set of positive semidefinite Hermitian matrices,
\emph{i.e.,} $\iota_{\PSD}(\mtx X) = 0$ if $\mtx X\in\PSD$ and
$\iota_{\PSD}(\mtx X) = +\infty$ if $\mtx X\notin \PSD$. Then,
equation~\eqref{eq:regularized} is equivalent to the composite
convex minimization problem
\begin{equation}
\minimize[\mtx X\in\SN]\ h(\mtx X) = f(\mtx X) + g(\mtx X)\text{.}
\label{eq:regularized_new}
\end{equation} 
One representative method to solve \eqref{eq:regularized_new} is the
so-called accelerated proximal gradient (APG) method
\cite{beck2009fast}. Given a function $g$ and $\alpha>0$, define the
proximal operator of $g$ as
\begin{equation*}\label{def:prox}
\Prox_{g}^{\alpha}\parens{\mtx X} = \argmin_{\mtx Y} \bigg\{g\parens{\mtx Y} + \tfrac{1}{2\alpha}\norm{\mtx Y - \mtx X}^2_\frobenius\bigg\}
\end{equation*}
where $\norm{\cdot}_\frobenius$ is the Frobenius norm.  Then, the APG
method constructs two sequences $\{\mtx X_k\}$ and $\{\mtx Y_k\}$ via
the following steps:
\begin{subequations}
	\begin{align}
	& t_k = \parens*{\sqrt{4(t_{k-1})^2+1}+1}/2,  \label{apg:1}\\
	& \mtx Y_k  =  \mtx X_k+\tfrac{t_{k-1}-1}{t_k}\parens{\mtx X_k-\mtx X_{k-1}}, \label{apg:2}\\
	& \mtx X_{k+1}  =  \Prox_{g}^{\alpha_k}\parens{\mtx Y_k - \alpha_k\nabla f(\mtx Y_k)}, \label{apg:3}
	\end{align}
\end{subequations}
where $t_0 = 1$, $0<\alpha_k\leq 1/L$, and $L$ is the maximal
eigenvalue of $\mathscr{A}^\HT\mathscr{A}$ with $\mathscr{A}^\HT$
being the adjoint operator for $\mathscr{A}$. With $\mathcal{X}_\star$
being the set of minimizers of \eqref{eq:regularized_new} and
$h_\star$ being the minimal value, it is well known that the APG
method has the following convergence property~\cite{beck2009fast}:
\begin{theorem}\label{thm:conv:apg}
	Let $\mtx X_k$ be the sequence generated by \eqref{apg:1}-\eqref{apg:3} and  $\alpha_k=1/L, \forall k$. Then for any $k\geq1$
	\begin{equation*}
	h(\mtx X_k) - h_\star\leq\frac{2L\dist(\mtx X_0,\mathcal{X}_\star)^2}{(k+1)^2},
	\end{equation*} 
	where $\dist(\mtx X,\mathcal{X}_\star)=\inf\{\norm{\mtx X-\mtx
          Y}_\frobenius:\mtx Y\in\mathcal{X}_\star\}$.
\end{theorem}

\subsection{The proposed algorithm} Indeed, the attractive $O(1/k^2)$ convergence rate of the APG method
is optimal when solving a convex problem when only first order
information is available~\cite{nesterov2013introductory}. However, it
has been observed that the objective value (the value of $f$ in this
case) sequence generated by the APG method shows oscillations, which
slows down convergence speed in practice~\cite{Odonoghue2015}.  To
avoid these oscillations, restarting techniques were proposed
in~\cite{Odonoghue2015} wherein $t_k$ is reset to $1$ when certain
criteria are met.  Despite the numerical success of this adaptive
restart APG method, it is still unknown whether this method is
guaranteed to converge.  We therefore propose a new restart criterion
that leads to a globally convergent numerical algorithm for solving
the trace-regularized coherence retrieval
problem~\eqref{eq:regularized}.  Given per-iteration step size
$\alpha_k>0$, define
\begin{equation}\label{prox:Z}
\mtx Z_{k+1} = \Prox_{g}^{\alpha_k}(\mtx Y_k - \alpha_k\nabla f(\mtx
Y_k)),\quad\mtx U_k = \mtx Y_k - \mtx Z_{k+1},\quad\mtx V_k =\mtx X_k - \mtx Z_{k+1}\text{.}
\end{equation}
When $\mtx X_k\neq \mtx Y_k$, we restart our APG method by setting
$t_k=1$ if the following does not hold:
\begin{equation}\label{res:grad}
\innerproduct{\mtx U_k}{\mtx V_k} - \alpha_k\innerproduct{\A{\mtx U_k}}{\A{\mtx V_k}}\geq \gamma \norm{\mtx V_k}_\frobenius^2
\end{equation}
where $\gamma>0$ is a small constant. We note that no extra computation is needed for checking
criterion \eqref{res:grad} due to the linearity of $\mathscr{A}$ and
the necessary quantities having been computed either during step size
estimation or in a previous iteration.  We summarize our proposed
adaptive APG algorithm in \cref{alg:adaAPG}.

\begin{remark}
$\Prox_g^\alpha\parens{\mtx
  X}$ is equal to a projection onto $\PSD$, independent of the value
of $\alpha$: 
\begin{equation*}
\Prox_g^\alpha\parens{\mtx X} = \Proj_{\PSD}\parens{\mtx X} =
\sum_{i=1}^N\nolimits \max\parens{\lambda_i,0} \vec q_i\vech q_i \text{,}
\end{equation*}
where $\set{\lambda_i}_i$ and $\set{\vec q_i}_i$ are the eigenvalues
and eigenvectors, respectively, of $\mtx X$.  The gradient of
quadratic function $f\parens{\mtx X}$ is equal to:
\begin{equation*}
\Grad f\parens{\mtx X} = \mathscr A^\HT\brackets[\big]{\mathscr A\parens*{\mtx X}-\vec
  b} + \mu\mtxh R\text{.}
\end{equation*}
In practice, the majority of computation is spent on evaluating
$\mathscr A$ and $\mathscr A^\HT$ as well as the eigenvalue
decomposition in the proximal operator.  Eigenvalue decomposition
scales asymptotically as $\mathcal{O}\parens*{N^3}$, whereas
evaluating $\mathscr A$ and $\mathscr A^\HT$ scales as
$\mathcal{O}\parens*{N^2M}$ in the worst case, with
$\mathcal{O}\parens*{N^3}$ being the best possible complexity via
reduced measurements~\cite{Chen2013} or exploiting the structure of $\mathscr
A$ using fast Fourier transforms or separable tensor
contractions~\cite{Shang2017}.
\end{remark}

\begin{algorithm}[ht]
	\caption{Adaptive APG algorithm}
	\label{alg:adaAPG}
	\begin{algorithmic}[1]
		\State Initialize $\mtx X_1=\mtx Y_1=\mtx Y_0=\mtx X_0$, $t_1=1$, $ k_{\text{max}}, k_{\text{maxres}}\in\mathbb{N}$, $k=1$, $k_{\text{res}}=0$, and  $\rho\in(0,1)$
		\While{$k\leq k_{\text{max}}$}
		\State Estimate step size $\alpha_k$ using
                \cref{alg:estStep} and obtain $\mtx Z_{k+1} =\Prox_{g}^{\alpha_k}(\mtx Y_k-\alpha_k\Grad f(\mtx Y_k))$.
		\If{\textbf{\{} $\mtx X_k=\mtx Y_k$ \textbf{or}
                  \eqref{res:grad} holds \textbf{\}} and
                  $k-k_{\text{res}}\leq k_{\text{maxres}}$ }
		\State Set $\mtx X_{k+1} = \mtx Z_{k+1}$.
		\State Set $t_{k+1} = \tfrac12\parens*{\sqrt{4t_k^2+1}+1}$.
		\State Set $\mtx Y_{k+1}  =  \mtx X_{k+1}+\frac{t_k-1}{t_{k+1}}(\mtx X_{k+1}-\mtx X_k)$.
		\State Set $k=k+1$.
		\Else
                \State Reset $t_k=1$ and update $k_{\text{res}}=k$.
                \State Set $\mtx Y_k=\mtx X_k$.
		\EndIf
		\EndWhile
	\end{algorithmic}
\end{algorithm}

\subsubsection{Step size estimation}
Setting $\alpha_k=1/L$ might be too conservative when the Lipschitz
constant $L$ is large. We would like to adaptively choose $\alpha_k$
by first initializing with the Barzilai-Borwein (BB) method
\cite{barzilai1988two}; our choice of the form with the squared norm
in the denominator is motivated by numerical considerations---an inner
product is sensitive to cancellation errors and placing it in the
denominator can result in numerical instability.  After
initialization, we then use a standard backtracking technique and
adopt the step size $\alpha_k$ at $\mtx Y_k$ whenever the following
inequality holds:
\begin{equation}
h(\mtx Y_k) - h(\mtx Z_{k+1})\geq\delta \norm{\mtx Y_{k}-\mtx Z_{k+1}}_\frobenius^2,\label{est:1}
\end{equation}
where $\delta>0$ is a small constant. It is noted that \eqref{est:1}
holds whenever $\alpha_k\leq1/(L+\delta)$ and hence backtracking must
terminate. In practice, we find that BB initialization gives a good
estimate for the step size and that backtracking is rare.  A detailed
listing for this step size estimation algorithm is given in
\cref{alg:estStep}.

\begin{algorithm}[ht]
	\caption{Estimation of step size $\alpha_k$}
	\label{alg:estStep}
	\begin{algorithmic}[1]
		\State {\bf Inputs:} $\mtx X_k$, $\mtx Y_k$, $\mtx Y_{k-1}$, $\Grad f(\mtx Y_k)$, $\Grad f(\mtx Y_{k-1})$, $\rho<1$ and $\alpha_{\text{min}},\alpha_{\text{max}}>0$.
		\State {\bf Outputs:} Step size $\alpha_k$, proximal
                point $\mtx Z_{k+1}$
                \If{k=1}
                \State Initialize $\beta=\norm{\vec
                  b-\mathscr{A}\parens{\mtx Y_k}}_\frobenius^2 /
                \norm*{\mathscr{A}^\HT\brackets*{\vec b-\mathscr{A}\parens{\mtx Y_k}}}_\frobenius^2$.
                \Else
		\State Calculate $\mtx S_k = \mtx Y_k-\mtx Y_{k-1}$ and $\mtx T_k = \nabla f(\mtx Y_k)-\nabla f(\mtx Y_{k-1})$.
		\State Initialize $\beta = \abs{\innerproduct{\mtx S_k}{\mtx T_k}}/\norm{\mtx T_k}_\frobenius^2$.
                \EndIf
		\For{$j=1,2\ldots$}
		\State Calculate $\mtx Z_{k+1} =
                \Prox_{g}^{\beta}\brackets*{\mtx Y_k - \beta\Grad
                  f(\mtx Y_k)}$.
                \If{\textbf{\{} $\mtx X_k\neq \mtx Y_k$ \textbf{and}
                    \eqref{res:grad} fails to hold \textbf{\}}
                  \textbf{or} \eqref{est:1} holds \textbf{or} $\beta<\alpha_{\text{min}}$}
		\State \textbf{break}
		\Else
		\State Backtrack $\beta = \rho\beta$.
		\EndIf 
		\EndFor
                \State Set $\alpha_k=\min\brackets*{\max(\alpha_{\text{min}},\beta),\alpha_{\text{max}}}$.
	\end{algorithmic}
\end{algorithm}

\section{Convergence analysis} 
\label{section:convergence}
In this section, we focus on convergence analysis for
\cref{alg:adaAPG} including an analysis regarding global convergence
as well as on the convergence rate. Before proceeding, we first
introduce some notation and definitions to be used in the analysis.

\subsection{Notation and definitions} 
Denote $\mathcal{X}_\star$ to be the solution set of
\eqref{eq:regularized}.  Given a point $\vec x$ and $\epsilon>0$, we
define $\mathbb{B}(\vec x,\epsilon)=\set{\vec y \ssuchthat \norm{\vec
    x-\vec y}_2\leq\epsilon}$.  Given a set $\mathcal{X}$, the
relative interior of $\mathcal{X}$, denoted by
$\ri\parens{\mathcal{X}}$, is defined as
\begin{equation*}
\ri\parens{\mathcal X}:=\set{\vec
	x\in\mathcal{X}\ssuchthat\exists\epsilon>0, \mathbb{B}(\vec
	x,\epsilon)\cap\aff\parens{\mathcal X}\subseteq\mathcal X},
\end{equation*}
where $\aff\parens{\mathcal{X}}$ is the affine hull of
$\mathcal{X}$. Given a set $\mathcal{X}$ and a member $\vec x$, the
distance from $\vec x$ to $\mathcal{X}$, denoted by $\dist\parens{\vec
  x,\mathcal{X}}$, is defined as $\dist\parens{\vec x,\mathcal X}=
\inf\set{\norm{\vec x-\vec y}_2\ssuchthat\vec y\in\mathcal{X}}$. Given
a function $f$ and $b\in\mathbb{R}$, we define the sub-level set
$\brackets*{f\parens{\vec x}\leq b}$ to be $\set{\vec x\ssuchthat
  f(\vec x)\leq b}$.  We use $\partial f$ to denote the (limiting)
subgradient of $f$.

Let $\mathcal A$ and $\mathcal B$ be finite dimensional Euclidean
spaces and $\Gamma:\mathcal A\rightrightarrows\mathcal B$ be a
\emph{set-valued mapping}. The \emph{graph}, \emph{domain} and
\emph{inverse} of $\Gamma$ are defined by
\begin{equation*}
\gph(\Gamma):=\set{(\vec u,\vec v)\ssuchthat\vec v\in\Gamma(\vec u)}\text{, }
\dom(\Gamma):=\set{\vec u\ssuchthat\Gamma(\vec u)\neq
	\emptyset}\text{, }
\Gamma^{-1}(v):=\set{\vec u\ssuchthat\Gamma(\vec u)=\vec v}.
\end{equation*}
In the following context, we define a useful property for set-valueed
mappings.
\begin{definition}\label{def:subregular} A set-valued mapping
	$\Gamma:\mathcal A\rightrightarrows\mathcal B$ is said to be
	metrically subregular at $\vec{\bar x}\in\mathcal A$ for $\vec{\bar
		y}\in\mathcal B$ if $(\vec{\bar x},\vec{\bar y})\in\gph(\Gamma)$ and there exist $\kappa,\epsilon>0$ such that
	\begin{equation}\label{eq:subregular}
	\dist(\vec x,\Gamma^{-1}(\vec{\bar y}))\leq \kappa
	\dist(\vec{\bar y}, \Gamma(\vec x)),\ \forall \vec x\in\mathbb{B}(\vec{\bar x},\epsilon).
	\end{equation}
\end{definition}

\subsection{Global convergence} 
\label{subsection:globalconvergence}
We first show that \cref{alg:adaAPG} converges to a global minimum,
provided the generated sequence is bounded.
\begin{theorem}
	Let $\{\mtx X_k\}$ be the sequence generated by
        \cref{alg:adaAPG}. If $\{\mtx X_k\}$ is bounded, then $\{\mtx
        X_k\}$ converges to a global minimum of
        \eqref{eq:regularized}, denoted as $\mtx{\bar X}$.
	\label{thm:convergence}
\end{theorem}
\begin{proof}
	See Appendix~\ref{proof1}.
\end{proof}
\begin{remark}
The proof is based on the recently established Kurdyka-\L{}ojasiewicz
(KL) property~\cite{attouch2010proximal,bolte2014proximal}, which
provides a framework for typical descent algorithms whose generated
sequences are bounded. However, the boundedness condition might not
hold for semidefinite programming problems such as
\eqref{eq:regularized}.
\end{remark}
In coherence retrieval, we will show that the boundedness condition on
$\{\mtx X_k\}$ holds when the measurement operator $\mathscr{A}$
satisfies a certain mild condition. Let $\mtx A = \parens{\mtx{\hat
    K}_1,\ldots,\mtx{\hat K}_M}^\HT\in\complexmatrices{MN}{N}$ where
the $\mtx{\hat K}_m\text{s} \in \complexmatrices{N}{N}$ are such that $\mtx
K_m=\mtx{\hat K}_m\mtxh{\hat K}_m$ for all $m\in\set{1,\ldots,M}$.  We
make the following assumption on $\mtx A$:
\begin{assume}
	\label{ass:bound}
	The matrix $\mtx A$ has full column rank.
\end{assume}
In practice, \cref{ass:bound} always holds because
either the measurements are designed to ensure $\mtx A$ is full rank,
or a smaller set of basis functions are used to remove the ambiguity,
\emph{e.g.,} we can choose a larger sampling interval in the case of
$\sinc$ basis functions, or we can use the nonzero right singular
vectors of $\mtx A$ to set a new basis based on our original basis.
\begin{proposition}
	\label{prop:bound}
	Let $\{\mtx X_k\}$ to be the sequence generated by \cref{alg:adaAPG}. Suppose \cref{ass:bound} holds.  Then, $\{\mtx X_k\}$ is bounded.
\end{proposition}

\begin{proof}
	See Appendix~\ref{p2}.
\end{proof}
Combining \cref{thm:convergence} and \cref{prop:bound}, we obtain that
\cref{alg:adaAPG} is globally convergent under \cref{ass:bound}.
\begin{corollary}
Let $\{\mtx X_k\}$ to be the sequence generated by
\cref{alg:adaAPG}. Suppose \cref{ass:bound} holds.  Then, $\{\mtx
X_k\}$ converges to a global minimum of \eqref{eq:regularized},
denoted as $\mtx{\bar X}$.
\end{corollary}

\subsection{Convergence rate analysis}
In this section, we first impose a reasonable condition on the
solution set $\mathcal{X}_\star$. Using this condition, we prove that
\cref{alg:adaAPG} converges linearly.
\begin{assume}
\label{ass:rate}
We make the following assumptions on $\mathcal{X}_\star$: a)
$\mathcal{X}_\star\neq\emptyset$.  b) There exists an $\mtx
X_\star\in\mathcal{X}_\star$ satisfying
\begin{equation}\label{sub:ass}
\mtx 0\in\Grad f(\mtx X_\star)+\ri\parens*{\mathcal{N}_{\PSD}(\mtx X_\star)}
\end{equation}
where $\mathcal{N}_{\PSD}\parens{\mtx X}$ denotes the normal cone of
$\PSD$ at $\mtx X$.
\end{assume}
\begin{remark}
	The first order optimality condition of $h$ implies
	\begin{equation}\label{eq:fo}
	0\in\Grad f(\mtx X)+\mathcal{N}_{\PSD}(\mtx X),\ \forall \mtx X\in\mathcal{X}_\star.
	\end{equation}
	Condition \eqref{sub:ass} is slightly more restrictive than
        \eqref{eq:fo}, but \eqref{sub:ass} only needs to hold at one
        point of $\mathcal{X}_\star$. Moreover, from the proof of
        \cref{prop:sub}, one sufficient condition for ensuring
        \eqref{sub:ass} is that there exists some $\mtx
        X_\star\in\mathcal{X}_\star$ such that $\rank(\mtx
        X_\star)=N$, \emph{i.e.,} $\mtx X_\star$ is full rank.
\end{remark}
Based on recent work~\cite{cui2016asymptotic}, we ensure $\partial h$
is metrically sub-regular at any $\bar{\mtx X}\in\mathcal X_\star$ for
$\mtx 0$ in the next proposition.
\begin{proposition}\label{prop:sub}
Let $h=f+g$ be defined in \eqref{eq:fandg}. Suppose \cref{ass:rate}
holds.  Then, for any $\bar{\mtx X}\in\mathcal{X}_\star$, $\partial h$ is
metrically sub-regular at $\bar{\mtx X}$ for $0$.
\end{proposition}
\begin{proof}
	See \cref{p1}.
\end{proof}

\begin{remark}
	Proposition~3.2 in~\cite{cui2016asymptotic} is a
        characterization of metric sub-regularity for real symmetric
        positive semidefinite matrices. However, the analysis
        in~\cite{cui2016asymptotic} can easily be extended for
        Hermitian positive semidefinite matrices over $\mathbb{R}$.
\end{remark}

Now, we can establish local linear convergence for
\cref{alg:adaAPG} via the following:
\begin{theorem}\label{thm:linear}
	Let $h=f+g$ be defined in \eqref{eq:fandg} and the sequence $\{\mtx
        X_k\}$ be generated by \cref{alg:adaAPG}. Suppose \cref{ass:bound} and \cref{ass:rate} holds. Then, there exists some $\bar{\mtx
          X}\in\mathcal{X}_\star$ such that one of the following assertions
        holds:\\
		1. $\set{\mtx X_k}$ converges to $\bar{\mtx X}$ in finite steps.\\
		2. $\set{h(\mtx X_k)}$ and $\set{\mtx X_k}$ linearly
                converge to $h(\bar{\mtx X})$ and $\bar{\mtx X}$, respectively, \emph{i.e.,} there exist $c_1,c_2>0$, $w_1,w_2\in(0,1)$ and $k_\ell>0$ such that 
		\begin{equation*}
		h(\mtx X_k)-h(\bar{\mtx X})\leq c_1w_1^k\text{, and }\norm{\mtx X_k-\bar{\mtx X}}_\frobenius\leq c_2w_2^k,\ \forall k>k_\ell.
		\end{equation*}
\end{theorem}
\begin{proof}
	See \cref{proof2}.
\end{proof}

\section{Results}

\label{section:results}

We now apply our algorithm to two data sets from translation-only
one-dimensional phase-space tomography~\cite{Rydberg2007, Tian2012}.
One is a simulation with realistic noise of two coherent Gaussian
beams that are slightly decohered with respect to each other.  Another
is experimental data from \cite{Zhang2013} consisting of a
Schell-model source imaged by a single positive lens.  In both cases,
the data consists of intensity profiles captured at
\SI{250}{\micro\meter} axial intervals by a single row of pixels in a
\SI{3.2}{\micro\meter} pitch camera with wavelength $\lambda$ equal to
\SI{532}{\nano\meter}.

While our approach can be applied to more complicated tomographic
apparatuses~\cite{McAlister1995,Marks2000} as long as the amplitude
transfer function $K\parens{\vec p,\vec q_m}$ is known, we focus on
the translation-only one-dimensional phase-space tomography example as
it is well-studied and a good base from which to extrapolate insights.

With one-dimensional phase-space tomography, the optical field as well
as the mutual intensity are assumed to be constant along one of the
spatial axes, \emph{i.e.,} $U\parens{\vec p}$ can be written as a
one-dimensional function $U\parens{x}$, being constant along $y$, and
thus $J\parens{\vec p_1,\vec p_2}$ can also be written as
$J\parens{x_1,x_2}$.  We assume that the optical field can be Nyquist
sampled at intervals of $\Delta$ and has negligible energy when
$\abs{x}>N\Delta/2$.  Thus, we employ the following $\sinc$ basis
functions:
\begin{equation*}
\xi_n\parens{x} = \sqrt{\Delta}\sin\brackets[\big]{\pi
  \parens{x/\Delta-n+n_0}}\big/\brackets[\big]{\pi
  \parens{x/\Delta-n+n_0}}
\end{equation*} 
where $n=1,\dots,N$ and $n_0=(N+1)/2$.  

For translation-only one-dimensional phase-space tomography, the light
from the source propagates through free space towards the measurement
points $\vec r_m$, each of which can be fully specified by transverse
position $x_m$ and axial position $z_m$; a camera sensor is translated
to the various axial positions and intensity measurements are obtained
from the pixel values in a single row.  We assume $\Delta\gg\lambda$
and thus use the Fresnel diffraction integral to compute the forward
model $\vec k_m$s and thus the measurement operators $\mtx K_m$s:
\begin{equation*}
  \begin{aligned}
    \vec k_{m}[n] &= \int \xi_n\parens{x}
    \exp\brackets*{\frac{\ii2\pi}{\lambda
        z_m}\parens{x_m-x}^2}\Big/\sqrt{\ii\lambda z_m} \dd x\\ &=
    \frac{\exp\parens*{\alpha_{m,n}}}{2\sqrt{\ii\lambda
        z_m/\Delta}}\braces*{\erf\brackets*{-\sqrt{\alpha_{m,n}}+\frac{\sqrt{\ii\pi\lambda
            z_m}}{2\Delta}}-\erf\brackets*{-\sqrt{\alpha_{m,n}}-\frac{\sqrt{\ii\pi\lambda
            z_m}}{2\Delta}}},
  \end{aligned}
\end{equation*}
where $\ii$ is the imaginary constant,
$\sqrt\ii=\parens{1+\ii}/\sqrt{2}$,
$\erf\parens{\zeta}=\tfrac{2}{\sqrt\pi}\int_0^\zeta\exp\parens{-t^2}\dd
t$ is the error function, and $\alpha_{m,n} =
\ii\pi\parens{x_m-n\Delta}^2/\parens{\lambda z_m}$.  Since these $\vec
k_m$s are deterministic, $\mtx K_m=\vect k_m\vec k_m^\conj$.

The constant parameters we chose for our numerical algorithm were:
\begin{equation*}
\begin{aligned}
&\mtx X_0=\mtx 0&&\delta=10^{-8}&&\gamma=10^{-5}&&\rho=\tfrac12\\
&\alpha_{\text{min}}=10^{-8}&&\alpha_{\text{max}}=10^8&&k_{\text{max}}=1000&&k_{\text{maxres}}=250
\end{aligned}
\end{equation*}
For comparison, we also run standard proximal gradients (PG),
accelerated proximal gradients (APG) and adaptive restart accelerated
proximal gradients (restart APG) using publicly available code
accompanying \cite{Odonoghue2015}.  For instrumentation, we added code
to the official source code in order to record the time taken and
acceleration parameter at each iteration.  Furthermore, each algorithm
was run a second time wherein the value of $\mtx X$ at each iteration
was recorded in order to compute the value of the merit function at
each iteration.  All computations were performed using \texttt{MATLAB}
running on an Intel Xeon E5-2630 CPU.

\subsection{Simulated Data}

We simulate a mostly coherent sum of two parallel Gaussian beams with
coplanar waists, as given by the following mutual intensity function:
\begin{equation*}
\begin{aligned}
J\parens{x_1,x_2} &=
G\parens{x_1;x_0}G\parens{x_2;x_0}+G\parens{x_1;-x_0}G\parens{x_2;-x_0}\\
&\quad+\chi\brackets[\big]{G\parens{x_1;x_0}G\parens{x_2;-x_0}+G\parens{x_1;-x_0}G\parens{x_2;x_0}}
\end{aligned}
\end{equation*}
where $G\parens{x;a} =
\exp\brackets*{-\parens{x-a}^2/\parens{2\sigma^2}}$,
$x_0=\SI{64}{\micro\meter}$, $\sigma=\SI{32}{\micro\meter}$ and
$\chi=0.9$.  This partially coherent field, discretized into a
\num{51x51} mutual intensity matrix ($N=51$) using a sampling interval of
$\Delta=\SI{6.4}{\micro\meter}$, is then propagated to \num{201} axial
positions spaced \SI{250}{\micro\meter} apart.  For each axial
position, we consider \num{101} intensity point samples spaced
\SI{3.2}{\micro\meter} apart for the measurements.  This set of true
intensities is shown in \figref{fig:sim:stacks} under the heading
``noiseless''.  The smallest and largest singular values of matrix
$\mtx A$ as it is defined in \Cref{subsection:globalconvergence} were
$a_{\text{min}}=3.094$ and $a_{\text{max}}=7.925$, respectively.

\begin{figure}[htp]
\centering
\includegraphics[width=0.9\linewidth]{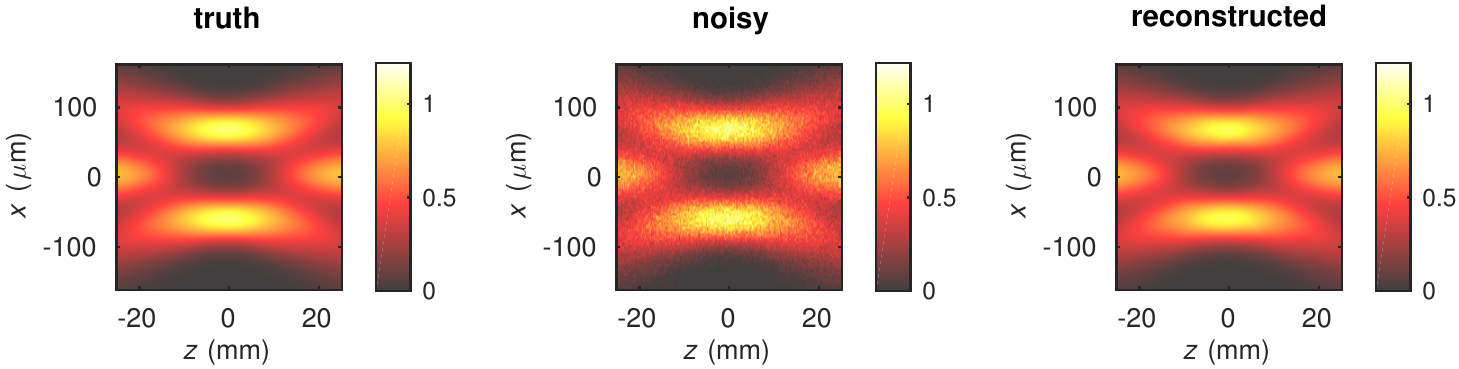}
\caption{Noisy measurements of the intensity are shown on left, with
  the true intensity shown on right.  Propagated intensity of a
  regularized reconstructed result is shown in center.  In all images,
  light propagates from left to right, and the intensity is in
  arbitrary units.\label{fig:sim:stacks}}
\end{figure}

To emulate real world conditions where the $\sigma_m$s are unknown, we
simulate collecting \num{16} measurements for each intensity point
sample, with their mean treated as the $y_m$s in our model and the
standard deviation across these \num{16} samples used as an
approximation to $\sigma_m\sqrt{16}$.  The noisy data for each
measurement is generated by first drawing from a Poisson distribution
whose rate parameter is proportional to the true intensity at that
point, with the sum of all the rate parameters made to equal
\SI{1.02e5} photons.  We then add Gaussian noise with standard
deviation equal to \num{0.01} times the maximum of all the rate
parameters; this is to simulate readout and quantization noise that
would be typically expected of an 8-bit sensor.  This set of simulated
measurements is shown in \figref{fig:sim:stacks} under the heading
``noisy''.  We then run our algorithm with the following inputs:
\begin{itemize}
\item \texttt{noiseless}:
The $y_m$s are set to the ideal noiseless intensity, with the $\sigma_m$s set to
  \num{1} and $\mu=0$ (\emph{i.e.,} no regularization).
\item \texttt{unregularized}: The $y_m$s and $\sigma_m$s are set to
  our simulated noisy data, with $\mu=0$ (\emph{i.e.,} no
  regularization).
\item \texttt{nuclear}: This uses the simulated noisy data and $\mtx
  R=\mtx I$ (\emph{i.e.,} nuclear norm regularization).  $\mu$ is set
  such that $\tfrac12\norm{\mathscr{A}\parens{\mtx X}-\vec
    b}_2^2\approx\alpha M/2$ upon convergence.  The $\alpha M/2$
  threshold for $\tfrac12\norm{\mathscr{A}\parens{\mtx X}-\vec b}_2^2$
  arises from the fact that the norm-squared expression follows a
  chi-squared distribution with $M$ degrees of freedom assuming that
  the $\sigma_m$s are correct values for the standard deviations of
  the Gaussian noise.  $M/2$ would be the mean for such a
  distribution, and $\alpha$ is used to adjust the threshold to take
  into account how well estimated the $\sigma_m$s are and how much of
  the distribution we want to include.  For this particular set of
  data, we are fairly confident of our estimates of $\sigma_m$ and
  have thus set $\alpha=1.5$.
\item \texttt{gradient}: This uses the simulated noisy data and $\mtx
  R=\mtx D$, a tridiagonal matrix with all the elements in the
  diagonal equal to $1$ and all the off-diagonal elements equal to
  $-\tfrac12$. $\mu$ is set such that
  $\tfrac12\norm{\mathscr{A}\parens{\mtx X}-\vec b}_2^2\approx\alpha
  M/2$ upon convergence.  The penalty term $\trace\parens{\mtxh D\mtx
    X}$ is thus equivalent to applying a Haar wavelet to both the rows
  and columns of $\mtx X$, keeping only the high frequency components
  and then taking the trace.  The application of $\mtx D$ is a simple
  approximation for a derivative operator on $x_1$ and $x_2$ in the
  continuous domain, and it physically corresponds to a desire to
  reduce the energy present in the first order spatial derivative of
  the optical wave function.  Mathematically, it penalizes nonsmooth
  solutions to our problem, and it is motivated by our \emph{a priori}
  knowledge that the true solution is smooth.
\item \texttt{early stop}: This is the same as the unregularized case,
  except we terminate our algorithm when
  $\tfrac12\norm{\mathscr{A}\parens{\mtx X}-\vec b}_2^2$ drops below
  $\alpha M/2$ where $\alpha=1.5$.  This result is used as a point of
  reference for comparison with the regularized results, since these
  results all have approximately the same level of measurement
  mismatch, and therefore the differences are due to the presence of
  and choice of regularizer.
\end{itemize}

\begin{figure}[htp]
\centering
\includegraphics[width=0.9\linewidth]{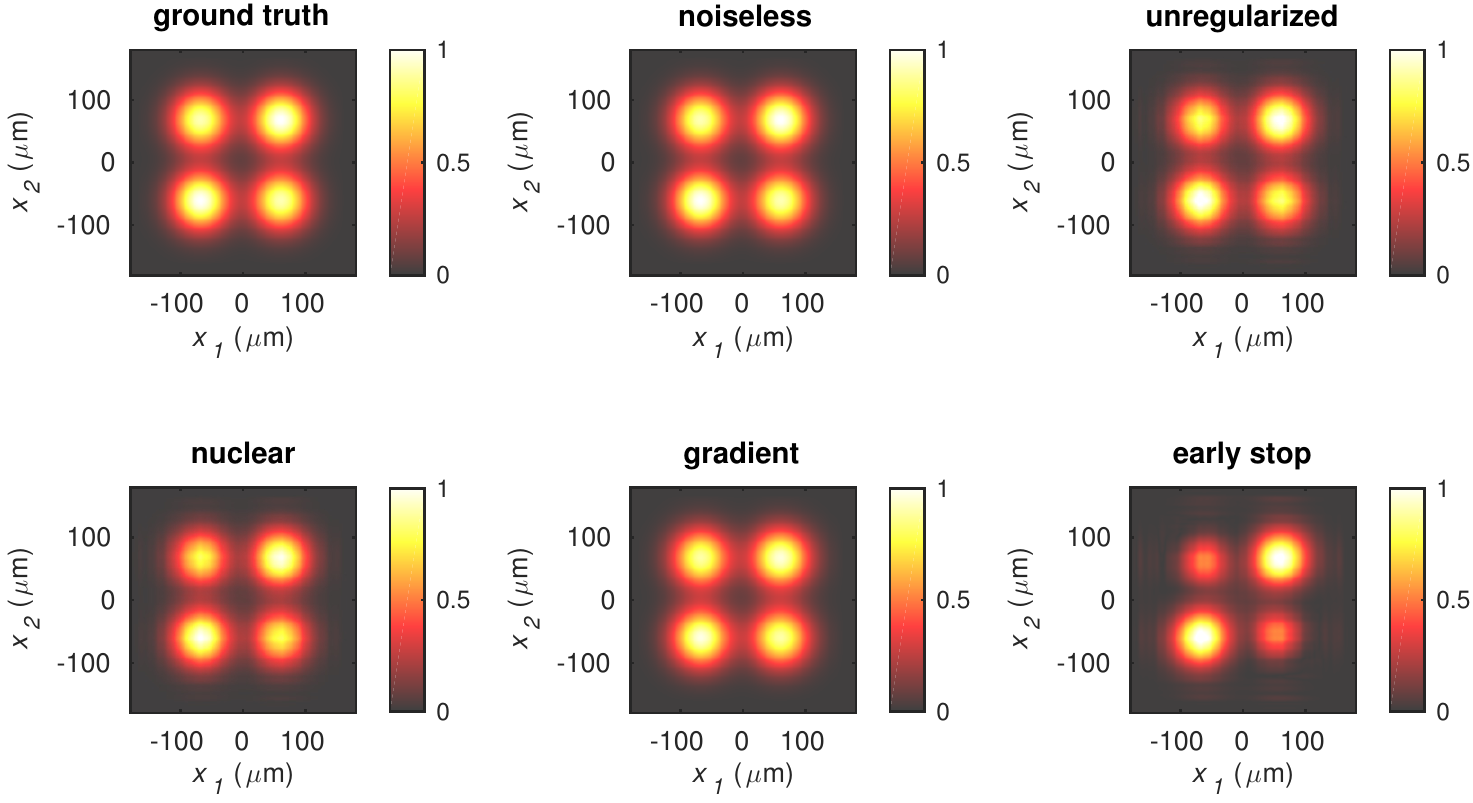}
\caption{Magnitude of the mutual intensity functions for the simulated
  data.  They are all drawn using the same scale, normalized to the
  maximum magnitude of the ground truth.\label{fig:sim:image}}
\end{figure}
\begin{figure}[htp]
\centering
\includegraphics[width=0.85\linewidth]{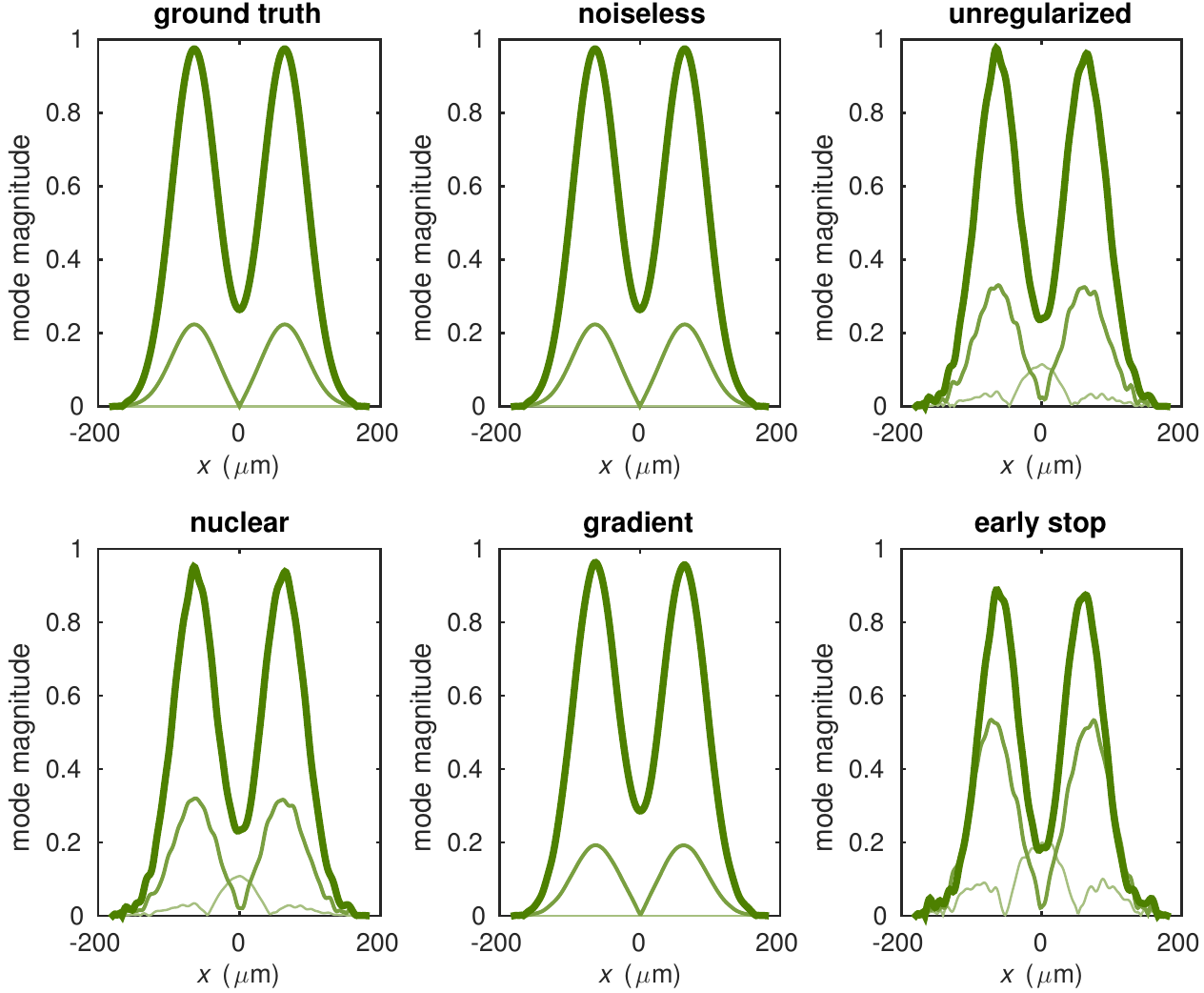}
\caption{Amplitude plots of the \num{3} highest energy
  coherent modes for each mutual intensity
  function.\label{fig:sim:modes}}
\end{figure}

\renewcommand{\arraystretch}{1.1}

\begin{table}[htp]
\centering
\footnotesize
\begin{tabular}{rccccc}
&
\texttt{noiseless}&
\texttt{unregularized}&
\texttt{nuclear}&
\texttt{gradient}&
\texttt{early stop}\\

\hline \\[-1.1em]

normalized error&
\num{7.524e-5}&
\num{0.1147}&
\num{0.1438}&
\num{0.03979}&
\num{0.3977}\\

trace distance&
\num{6.103e-5}&
\num{0.07908}&
\num{0.07658}&
\num{0.02534}&
\num{0.2887}\\

\hline \\[-1.1em]

\end{tabular}

\vspace{0.5\baselineskip}

\caption{Quantitative measurements of error for the different
  reconstructed mutual intensity matrices.  The normalized error is
  defined as $\norm{\mtx X-\mtx X_{\text{true}}}_\frobenius/\norm{\mtx
    X_{\text{true}}}_\frobenius$.  The trace distance, a quantity used
  to describe the difference between two quantum states, is equal to
  half the sum of the singular values of $\rho-\rho_{\text{true}}$,
  where $\rho$ and $\rho_{\text{true}}$ are equal to $\mtx
  X/\trace\parens{\mtx X}$ and $\mtx
  X_{\text{true}}/\trace\parens{\mtx X_{\text{true}}}$ respectively.\label{table:sim:error}}
\end{table}

\begin{figure}[htp]
\centering
\includegraphics[width=0.9\linewidth]{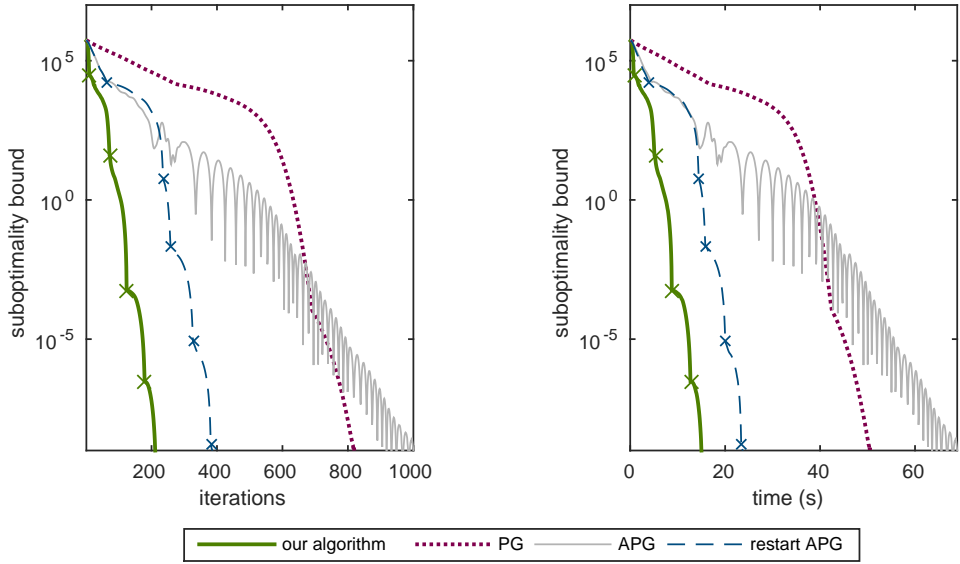}
\caption{Comparison of convergence across different algorithms for the
  simulated data.  The vertical axis is the difference between the
  value of $f\parens{\mtx X}$ and the lowest attained value of
  $f\parens{\mtx X}$ across all algorithms, which in turn gives an
  upper bound on the suboptimality.  On left, the horizontal axis is
  the number of iterations.  On right, the horizontal axis is time
  taken.  The $\times$s mark where restarts
  occurred.\label{fig:sim:algorithm}}
\end{figure}

The magnitude of the reconstructed mutual intensity functions are
shown in \figref{fig:sim:image}, and the magnitude of their coherent
modes~\cite{Wolf1982} are given in
\figref{fig:sim:modes}.\footnote{The $i$th coherent mode is
  $\sqrt{\lambda_i}\vect v_i\vec\xi\parens{x}$ where $\mtx X=\sum_i
  \lambda_i\vec v_i\vech v_i$ is $\mtx X$'s eigenvalue decomposition.}
The propagated intensity using the \texttt{gradient} result is shown
in \figref{fig:sim:stacks}.  Since the ground truth is known, we give
a summary of the reconstruction error using various metrics in
\tabref{table:sim:error}.  A convergence comparison between algorithms
for the \texttt{gradient} input is shown in
\figref{fig:sim:algorithm}.

As is evident from the results, especially \figref{fig:sim:modes},
noise introduces additional energy (as seen in the increased amplitude
for the second mode) as well as an additional mode, resulting in a
reconstructed mutual intensity that appears \emph{less} coherent than
the ground truth.  Furthermore, noise also induces additional high
frequency content in the reconstructed result.  For the specific value
of $\alpha$ used, regularization using the nuclear norm only yields
marginal improvements over the unregularized reconstruction, whereas
regularization using the smoothness-inducing regularizer both smooths
the result and reduces the number of significant modes back down to
the correct number of modes.  Of course, we can increase the $\mu$
parameter for the nuclear norm case to reduce the number of modes, but
it does not fully smooth the result and results in a mutual intensity
that does not match the measurements as well as the smooth-prior
regularized result.  We note that the \texttt{early stop} result has
the same amount of measurement mismatch as the two regularized
solutions, showing how ill-posed the problem is and the necessity for
regularization.

Our algorithm also converges faster than the three other methods given
in \figref{fig:sim:algorithm}, albeit restart APG converges
asymptotically as fast.  Non-restarting APG oscillates due to excess
critical momentum, as described in \cite{Odonoghue2015}.  The
advantages of our algorithm are that: (1) it is provably convergent,
whereas to the best of our knowledge no proof of convergence has been
given for the algorithm given in \cite{Odonoghue2015}, and (2) it
converges much faster than the provably convergent non-restarting APG
(\emph{i.e.,} FISTA~\cite{beck2009fast}).

\subsection{Experimental Data}

We use the experimental data from \cite{Zhang2013}, wherein the
intensity profile of a partially coherent beam is imaged at \num{201}
positions along the optical axis.  The partially coherent beam was
generated by focusing an LED light source through a
\SI{532}{\nano\meter} bandpass filter onto a \SI{100}{\micro\meter}
slit located at the front focal plane of a \SI{100}{\milli\meter}
focal length cylindrical lens.  A \SI{500}{\micro\meter} slit placed
at the back focal plane is illuminated by light passing through the
cylindrical lens, and this slit is imaged using a
\SI{50}{\milli\meter} cylindrical lens placed \SI{150}{\milli\meter}
after the slit.  Based on visual inspection, the axial positions
captured were located between $z=\SI{-30.25}{\milli\meter}$ and
$z=\SI{19.75}{\milli\meter}$ relative to the image of the slit.  

The experimental data $y_m$s are taken from a single row on a camera
with \SI{3.2}{\micro\meter} pitch, and the standard deviation
$\sigma_m$s were estimated from the neighboring \num{16} rows.  A
visualization of the measured and theoretical intensities are shown in
\figref{fig:exp:stacks}.  We use a $\sinc$ basis with sampling
interval $\Delta=\SI{6.4}{\micro\meter}$ to discretize the mutual
intensity into a \num{101x101} matrix ($N=101$).  We note that the
field does not necessarily have a spatial band-limit compatible with
the sampling interval, but since we only measure the intensity at
intervals of \SI{3.2}{\micro\meter}, it would be very difficult to
recover any information at a higher sampling rate and hence we ignore
the higher frequency components.  The smallest and largest singular
values of the matrix $\mtx A$ as it is defined in
\Cref{subsection:globalconvergence} were $a_{\text{min}}=7.625$ and
$a_{\text{max}}=14.18$.

An aberration-free theoretical estimate of the mutual intensity is:
\begin{equation*}
\begin{aligned}
J\parens{x_1,x_2} &\propto
\rect\parens{\beta_1x_1}\rect\parens{\beta_1x_2}
\sinc\brackets[\big]{\beta_2\parens{x_1-x_2}}
\exp\brackets[\big]{\ii\beta_3\parens{x_1^2-x_2^2}}\\
&\quad\otimes\sinc\parens{x_1/\Delta}\sinc\parens{x_2/\Delta}
\end{aligned}
\end{equation*}
where $\beta_1=\SI{4}{\per\milli\meter}$,
$\beta_2^{-1}=\SI{532}{\micro\meter}$,
$\beta_3=\SI{1.9684e-4}{\per\micro\meter\squared}$, and $\otimes$
denotes convolution.  Since ground truth is not available, we use this
aberration-free estimate as a rough point of reference; it is not
intended to be interpreted as the ground truth, which may be slightly
blurred or distorted by aberrations.

\begin{figure}[htp]
\centering
\includegraphics[width=0.9\linewidth]{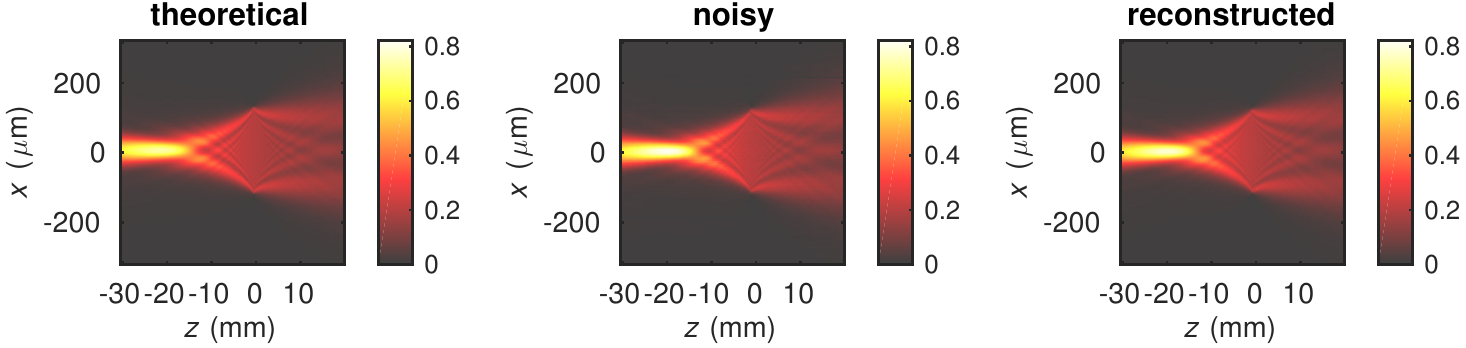}
\caption{Noisy measurements of the intensity are shown on left, with
  the theoretical aberration-free intensity profile shown on the
  right.  Propagated intensity of a regularized reconstructed result
  is shown in center.  In all images, light propagates from left to
  right, and the intensity is in arbitrary units.  The theoretical
  intensity is provided as a point of reference and is not necessarily
  the ground truth.\label{fig:exp:stacks}}
\end{figure}

We again use several different sets of parameters for our algorithm,
although we replaced two of the input sets and used a different value
of $\alpha$ for the target value of
$\tfrac12\norm{\mathscr{A}\parens{\mtx X}-\vec b}_2^2$ in the
regularized reconstructions.  Instead of the \texttt{early stop} and
\texttt{noiseless} input parameter sets, we added the
\texttt{nuclear+support} and \texttt{window} input parameter sets.
The former uses the same regularizer as the \texttt{nuclear} data set,
but additionally forces coefficients of basis functions whose centers
lie outside the center \SI{250}{\micro\meter} region to be zero, as a
way to demonstrate the state of the art nuclear norm regularizer
combined with a \emph{hard} support constraint.  The \texttt{window}
input parameter set uses $\mtx R=\mtx W$, a diagonal matrix whose
entries are unity in the center \SI{250}{\micro\meter} region and
increase linearly away from this region, up to a maximum of \num{391}
at the ends.  The idea with these additional input sets is to
demonstrate how one can incorporate \emph{a priori} information about
the support of the solution---since we know our slit should be imaged
to a region that wide, we would like to penalize any contributions
outside of this region.  While \texttt{nuclear+support} uses a hard
constraint, it is not necessarily appropriate because the field may
not actually be exactly zero outside the region, due to the presence
of possible aberrations in the system.  The \texttt{window} approach
is a gentler way of finding a less energetic solution while at the
same time preferring to remove energy from areas where we do not
expect much energy to be present, \emph{i.e.,} it is a \emph{soft}
support constraint.  Coincidentally, this regularizer is also
equivalent to imposing a smoothness constraint on the intensity in the
far field of the partially coherent beam.  We also chose a value of
$\alpha=5$ to account for additional possible errors in $\mathscr{A}$
(due to imperfect equipment and calibration) and standard deviation
estimation.

\begin{figure}[htp]
\centering
\includegraphics[width=0.9\linewidth]{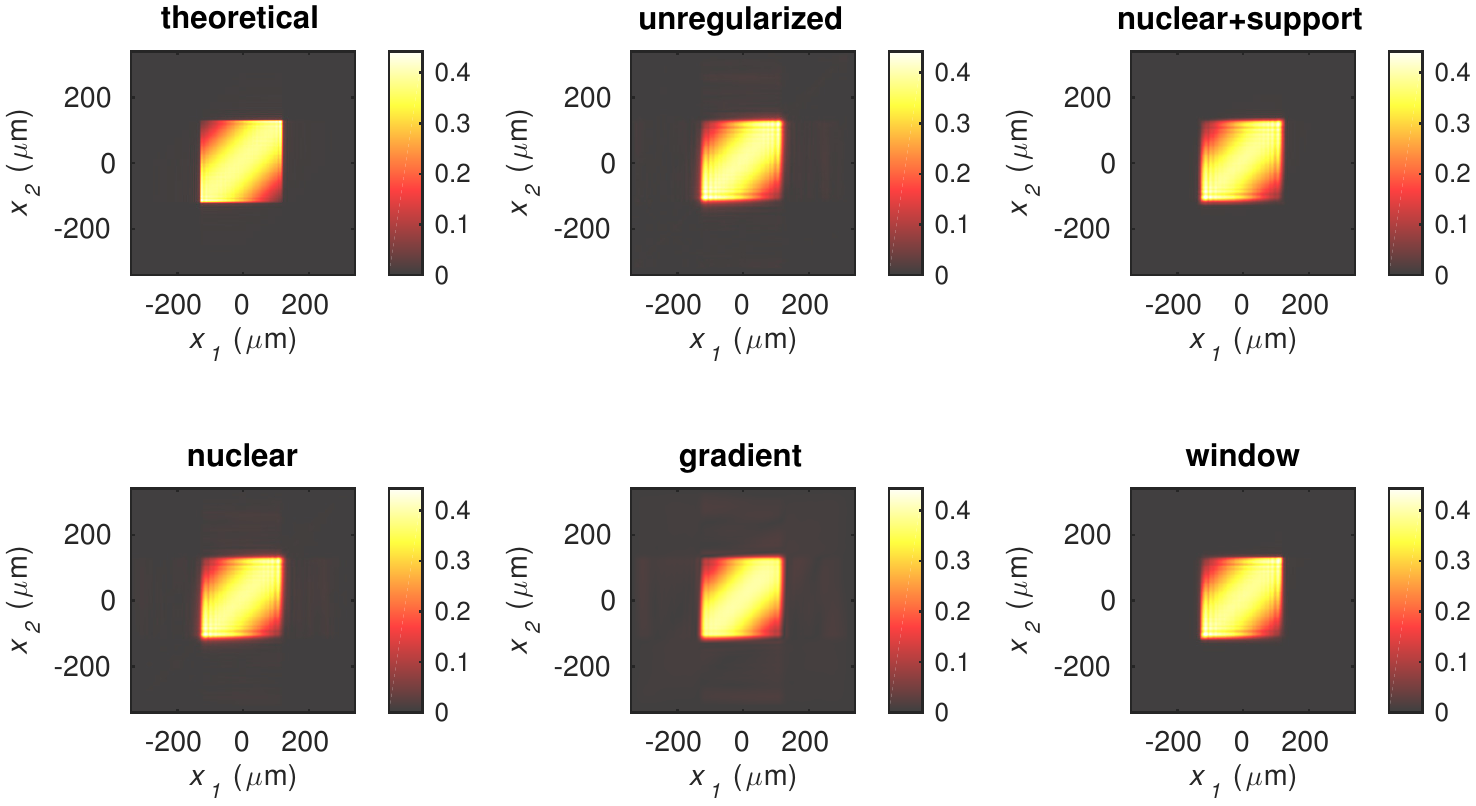}
\caption{Magnitude of the mutual intensity functions for the
  experimental data.  They are all drawn using the same scale, with
  values in arbitrary units.\label{fig:exp:image}}
\end{figure}
\begin{figure}[htp]
\centering
\includegraphics[width=0.85\linewidth]{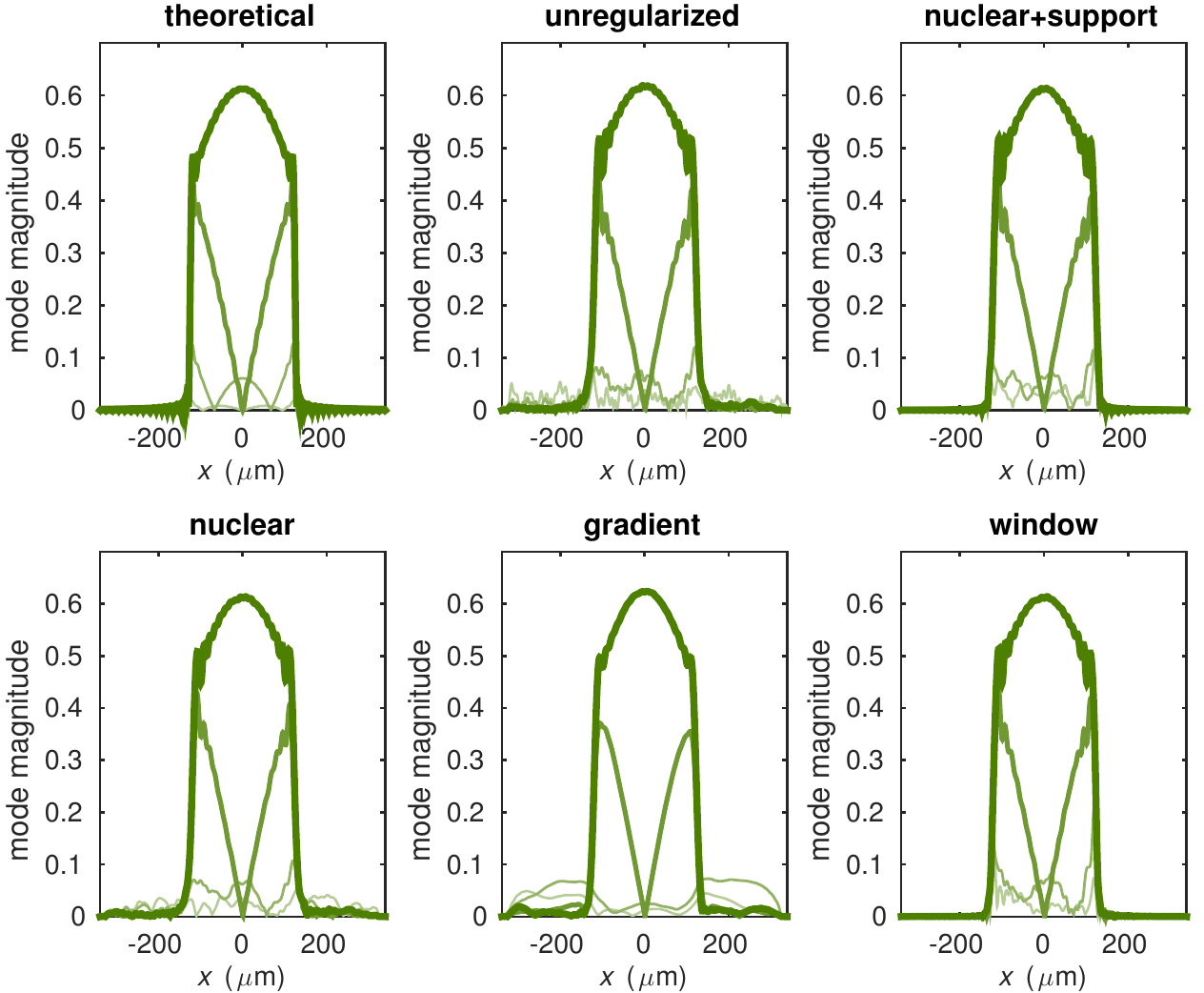}
\caption{Magnitude plots of the \num{4} highest energy
  coherent modes for each mutual intensity
  function.\label{fig:exp:modes}}
\end{figure}

\begin{table}[htp]
\footnotesize
\centering
\begin{tabular}{rccccc}

&
\texttt{unregularized}&
\texttt{nuclear+support}&
\texttt{nuclear}&
\texttt{gradient}&
\texttt{window}\\

\hline\\[-1.1em]

normalized RMSE&
\num{0.2363}&
\num{0.2336}&
\num{0.2317}&
\num{0.2189}&
\num{0.2176}\\

trace distance&
\num{0.1795}&
\num{0.1717}&
\num{0.1666}&
\num{0.1743}&
\num{0.1486}\\

\hline\\[-1.1em]

\end{tabular}
\vspace{0.5\baselineskip}
\caption{Quantitative measurements of the difference between the
  mutual intensity matrices reconstructed from experimental data and
  the aberration-free mutual intensity estimate.  The quantities are
  defined in the same way as \cref{table:sim:error}, with the
  aberration-free mutual intensity estimate taken as the
  ``truth''.\label{table:exp:error}}
\end{table}

\begin{figure}[htp]
\centering
\includegraphics[width=0.9\linewidth]{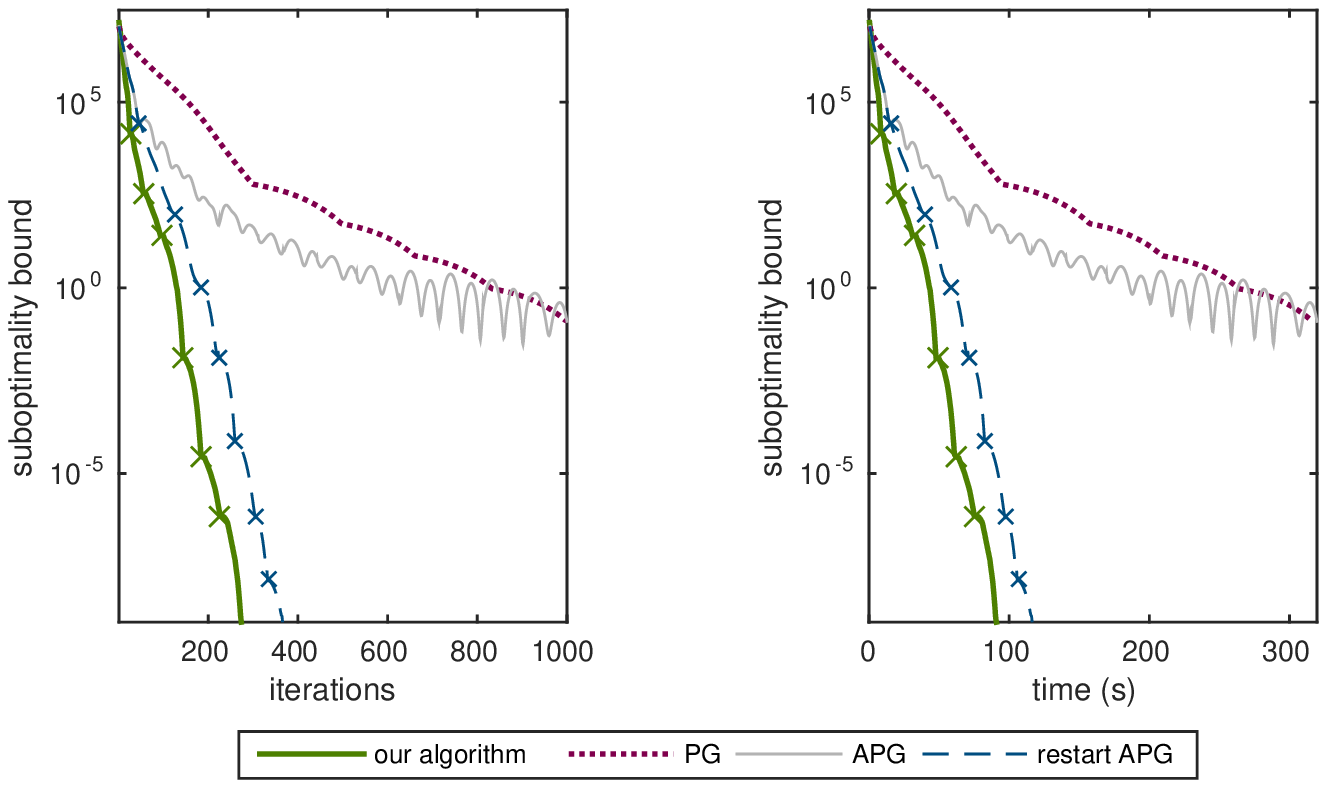}
\caption{Comparison of convergence across different algorithms for the
  experimental data.  The vertical axis is the difference between the
  value of $f\parens{\mtx X}$ and the lowest attained value of
  $f\parens{\mtx X}$ across all algorithms, which in turn gives an
  upper bound on the suboptimality.  On left, the horizontal
  axis is the number of iterations.  On right, the horizontal axis
  is time taken.  The $\times$s mark where restarts
  occurred.\label{fig:exp:algorithm}}
\end{figure}

The reconstructed mutual intensity functions are shown in
\figref{fig:exp:image} with modes shown in \figref{fig:exp:modes}.  We
give quantitative comparisons between the theoretical mutual intensity
and the reconstructed mutual intensity in \tabref{table:exp:error},
and a comparison of algorithm convergence for the \texttt{window}
input set is given in \figref{fig:exp:algorithm}.  All regularized
solutions remove some amount of noise, as evident in the increased
smoothness and reduction of high order mode energy in the coherent
modes visualization.  The \texttt{nuclear} result cleans up the
reconstruction compared to the noisy reconstruction, but it still
leaves a lot of excess energy in the higher order modes, especially
energy outside the region occupied by the slit.  The \texttt{gradient}
regularizer is good at reducing the number of modes, but it
oversmoothes the result---the third mode spills outside of the region
occupied by the slit and resembles neither the third mode from the
theoretical results nor that of the noisy results, and the sharp edges
of the second mode are gone.  The image of the magnitude
of the mutual intensity in \figref{fig:exp:image} is also quite
blurry.  While the \texttt{nuclear} result is an improvement over the
unregularized result, it is obvious that applying additional prior
information about the support yields a much better reconstruction, as
can be seen in the \texttt{nuclear+support} and \texttt{window}
cases.  However, the application of a hard support constraint might
not be suitable in this particular situation, as we do not know about
the extent of aberrations in the imaging system.  Furthermore,
\texttt{window} does manage to reconstruct the third mode better than
all of the other methods; \texttt{nuclear+support} still does not
perform as well, and leaves excess energy in the fourth mode.

The experimental data convergence results in
\figref{fig:exp:algorithm} are quite similar to the one for the
simulated data---our algorithm has an asymptotic convergence rate
comparable to gradient restart APG, albeit we again reach the fast
convergence regime faster.

\section{Conclusion}

\label{section:discussion}

We have demonstrated that trace-regularized coherence retrieval can be
a powerful tool in recovering the mutual intensity when the inverse
problem is ill-conditioned.  The generalization of the nuclear-norm
enables flexibility in applying \emph{a priori} information, leading
to higher quality reconstructions.  Furthermore, we have demonstrated
an efficient numerical scheme for our coherence retrieval model, with
performance at worst comparable to the state-of-the-art adaptive
restart APG scheme while simultaneously being provably globally
convergent, with mild conditions required for linear convergence.

This work uses very simple $\mtx R$ matrices for regularization, with
good results, but more flexibility and power can be attained by
leveraging the framework of tight frames through further study.
Furthermore, a method for exploiting redundant information in the
measurements to calibrate real-world $\mtx K_m$s as well as methods to
reduce the memory and computational footprint for high dimensional
structured mutual intensity matrices are all potential avenues for
future exploration.

\bibliographystyle{siamplain}
\bibliography{regcohret}

\appendix

\section{Proof of \cref{thm:convergence}}
\label{proof1}
Given
$\eta\in(0,+\infty]$, define $\Phi_\eta$ be the class of all concave
and continuous functions $\phi:[0,\eta)\rightarrow\mathbb{R}_+$ that
satisfy $\phi(0)=0$, $\phi$ is $C^1$ on $(0,\eta)$ and continuous at
$0$, and $\phi^{'}(s)>0,\ \forall s\in(0,\eta)$.  
\begin{definition}\label{def:KL}
	Let $f:\mathbb{R}^n\rightarrow(-\infty,+\infty]$ be proper and
	lower semicontinuous. The function $f$ is said to satisfy the
	{\emph{Kurdyka-\L{}ojasiewicz}} (KL) inequality at $\vec{\bar
		x}\in\dom\parens{\partial f}$ if there exist $\eta\in(0,+\infty]$, a
	neighborhood $\mathcal U$ of $\vec{\bar x}$ and a function
	$\phi\in\Phi_\eta$, such that for all $\vec x\in\mathcal
	U\cap\brackets*{f(\vec{\bar x})<f(\vec x)<f(\vec{\bar x})+\eta}$, the following inequality holds:
	\begin{equation}\label{eq:KL}
	\phi^{'}(f(\vec x)-f(\vec{\bar x}))\dist(\vec 0,\partial
	f(\vec x))\geq 1.
	\end{equation}
	A function $f(\vec x)$ is called a KL function if $f$
	satisfies the KL property at every $\vec x\in\dom\parens{\partial f}$.
\end{definition}
We will now show some basic properties
of \cref{alg:adaAPG} in the following lemmas and use them to prove the
global convergence of \cref{alg:adaAPG}.
\begin{lemma}\label{lemma:1}
	Let $\{\mtx X_k\}$ be the sequence generated by
	\cref{alg:adaAPG}. Then, there exists $a,b>0$ and $\bar w\in[0,1)$ such that for all $k\geq 0$, we have
	\begin{align}
		h(\mtx X_k)-h(\mtx X_{k+1}) &\geq a\norm{\mtx X_k-\mtx X_{k+1}}_\frobenius^2, \label{le:suff} \\
		{\dist}(0,\partial h(\mtx X_{k+1})) &\leq b\parens*{\norm{\mtx X_{k+1}-\mtx X_k}_\frobenius+\bar w\norm{\mtx X_k-\mtx X_{k-1}}_\frobenius}.\label{le:diff}
	\end{align}
\end{lemma}
\begin{proof}
	Define 
	\begin{equation}\label{prox:W}
	\mtx W_{k+1} =\Prox_{g}^{\alpha_k}(\mtx X_k-\alpha_k\Grad f(\mtx X_k)).
	\end{equation}
	Let $\Omega_1 =\set{k\ssuchthat\mtx X_{k+1}=\mtx W_{k+1},
		k>1}$ and $\Omega_2=\set{k\ssuchthat\mtx X_{k+1}=\mtx
		Z_{k+1}}$ where $\mtx Z_{k+1}$ is defined in \eqref{prox:Z}. Then, $\Omega_1\cap\Omega_2=\emptyset$ and
	$\Omega_1\cup\Omega_2=\mathbb{N}$. We consider two cases:
	(I) $k\in\Omega_1$ and (II) $k\in\Omega_2$.
	
	\noindent {\bf Case I.} $k\in\Omega_1$ implies that $\mtx X_k=\mtx
	Y_k$, and hence $h(\mtx X_k)-h(\mtx X_{k+1})\geq \delta\norm{\mtx X_k-\mtx
		X_{k+1}}_\frobenius^2$ due to inequality \eqref{est:1} from the step size
	estimation process. From the optimality condition of \eqref{prox:W},
	we know
	\begin{equation}\label{diff:sub1}
	\tfrac{1}{\alpha_k}(\mtx X_k-\mtx X_{k+1})-\Grad f(\mtx X_k)\in\partial g(\mtx X_{k+1}).
	\end{equation}
	The  inequality \eqref{diff:sub1} implies
	\begin{equation*}
		\begin{aligned}
			\dist(0,\partial h(\mtx X_{k+1}))& \leq \norm*{\Grad f(\mtx
				X_{k+1}) 	+\tfrac{1}{\alpha_k}(\mtx X_k-\mtx
				X_{k+1})-\Grad f(\mtx X_k)}_\frobenius \\
			& \leq (L+1/\alpha_{\text{min}}) \norm{\mtx X_{k+1}-\mtx X_k}_\frobenius.
		\end{aligned}
	\end{equation*}
	since $\alpha_k\geq \alpha_{\text{min}}>0,\forall k$.\\
	{\bf Case II: $k\in\Omega_2$.} From the first order optimality condition of \eqref{prox:Z}, we have
	\begin{equation}\label{diff:sub2}
	\tfrac{1}{\alpha_k}(\mtx Y_k-\mtx Z_{k+1})-\Grad f(\mtx Y_k)\in\partial g(\mtx Z_{k+1}).
	\end{equation}
	Then, together with the convexity of $h$ and the fact that
	$\mtx X_{k+1}=\mtx Z_{k+1}$, we have
	\begin{equation*}
		\begin{aligned}
			h(\mtx X_{k})-h(\mtx X_{k+1})&\geq \innerproduct*{\Grad f(\mtx Z_{k+1}) 	+\tfrac{1}{\alpha_k}(\mtx Y_k-\mtx Z_{k+1})-\Grad f(\mtx Y_k)}{\mtx X_{k}-\mtx Z_{k+1}}\\
			&=
			\innerproduct*{\parens*{\tfrac{1}{\alpha_k}\mathbf{I}-\mathscr{A}^\HT\mathscr{A}}(\mtx
				Y_k - \mtx Z_{k+1})}{\mtx X_k - \mtx Z_{k+1}}\\
			&  \geq \frac{\gamma}{\alpha_{\text{max}}}\norm{\mtx X_k-\mtx X_{k+1}}_\frobenius^2,
		\end{aligned}
	\end{equation*}
	since the inequality \eqref{res:grad} holds and
	$0<\alpha_{\text{min}}\leq\alpha_k\leq\alpha_{\text{max}}$. Moreover, the inequality
	\begin{equation*}
		\begin{split}
			\dist\parens*{0,\partial h(\mtx X_{k+1})}& \leq \norm*{\Grad f(\mtx Z_{k+1}) 	+\tfrac{1}{\alpha_k}(\mtx Y_k-\mtx Z_{k+1})-\Grad f(\mtx Y_k)}_\frobenius \\
			& \leq (L+1/\alpha_{\text{min}}) \parens*{\norm{\mtx X_{k+1}-\mtx X_k}_\frobenius+\bar w\norm{\mtx X_k-\mtx X_{k-1}}_\frobenius}
		\end{split}
	\end{equation*}
	holds where $\bar w\in[0,1)$ since $(t_k-1)/t_{k+1}<1$ for all $k\leq k_{\text{maxres}}$.
	
	Consequently, the two cases yield that the inequality
	\eqref{le:suff} holds with $a=\min(\delta,\gamma/\alpha_{\text{max}})>0$ and the
	inequality \eqref{le:diff} holds with $b=L+1/\alpha_{\text{min}}>0$.
\end{proof}
Denote $w(\mtx X_0)$ to be the limiting points of $\{\mtx X_k\}$. The
next lemma shows that \cref{alg:adaAPG} is subsequence convergent,
\emph{i.e.,} all the convergent subsequences converge to a minimal
point when the generated sequence is bounded.
\begin{lemma}\label{lemma:2}
	Let $\{\mtx X_k\}$ be the sequence generated by
	\cref{alg:adaAPG} starting from $\mtx X_0$. If $\{\mtx X_k\}$ is bounded, then $w(\mtx X_0)$ is a non-empty compact set, and $w(\mtx X_0)\subseteq\mathcal{X}_\star\neq\emptyset$.
\end{lemma}
\begin{proof}
	Since the sequence $\{\mtx X_k\}$ is bounded,
	$w(\mtx X_0)$ is
	non-empty. Meanwhile, $w(\mtx X_0)$ is compact since it is the
	intersection of compact sets, \emph{i.e.,} $w(\mtx X_0)
	=\cap_{q\in\mathbb{N}}\overline{\cup_{k\geq q}\{\mtx X_k\}}$,
	where $\overline{\mathcal A}$ denotes the closure of set
	$\mathcal A$. Let $\bar{\mtx X}$ be any point in $w(\mtx X_0)$
	and $\{\mtx X_{k_j}\}$ be a convergent subsequence of $\{\mtx
	X_k\}$ such that $\mtx X_{k_j}\rightarrow\bar{\mtx X}$ as
	$j\rightarrow+\infty$.  Since $\mtx R\in\PSD$, $h(\mtx X)\geq
	0,\forall \mtx X$. Together with the inequality
	\eqref{le:suff}, we know there exists some $\bar h$ such that
	$h(\mtx X_k)\rightarrow\bar h$ as
	$k\rightarrow+\infty$. Moreover, \eqref{le:suff} implies
	\begin{equation*}
		a\sum_{k=0}^{\infty}\norm{\mtx X_k-\mtx X_{k+1}}_\frobenius^2\leq \sum_{k=0}^{\infty} (h(\mtx X_{k})-h(\mtx X_{k+1}))\leq h(\mtx X_0)-\bar h<+\infty.
	\end{equation*}
	So, we have $\norm{\mtx X_{k_j+1}-\mtx
		X_{k_j}}_\frobenius\rightarrow0$ and $\norm{\mtx
		X_{k_j}-\mtx X_{k_j-1}}_\frobenius\rightarrow0$ as
	$k\rightarrow+\infty$.  Moreover, from the above facts, it is
	easy to prove that $\{\mtx X_{k_j+1}\}$ also converges to
	$\bar X$. Together with \eqref{le:diff}, we know $\dist(\mtx
	0,\partial h(\mtx X_{k_j+1}))\rightarrow0$ as
	$j\rightarrow+\infty$. Let $\mathcal U$ be a neighborhood of $\bar{\mtx X}$ with radius $M$. Then, we have
	\begin{equation}\label{limits:subdiff}
	h(\mtx X)\geq h(\mtx X_{k_j+1})-\dist(\mtx 0,\partial h(\mtx
	X_{k_j+1}))\norm{\mtx X-\mtx X_{k_j+1}}_\frobenius,\ \forall
	\mtx X\in \mathcal U.
	\end{equation}
	Since $\lim\limits_{j\rightarrow+\infty}h(\mtx X_{k_j+1})=
	h(\bar{\mtx X})$, taking the limit $j\rightarrow+\infty$ in
	\eqref{limits:subdiff}, we have $h(\mtx X)\geq h(\bar{\mtx
		X})$ for all $\mtx X\in \mathcal U$.  By the convexity of
	$h$, $\bar{\mtx X}$ is a global minima of
	\eqref{eq:regularized}, i.e.\ $\mtx 0\in\partial h(\bar{\mtx X})$.
\end{proof}
Based on the proof of Theorem~1 in \cite{bolte2014proximal}, we can prove \cref{thm:convergence} as follows.

\begin{proof}
	It is noted that the objective function $h$ in \eqref{eq:regularized} is a KL function as both $f$ and $g$ are semi-algebraic. Since $\{\mtx X_k\}$ is bounded, $w(\mtx X_0)$ is not
	empty. Denote $h(\mtx X)=\bar{h}$ for all $\mtx X\in w(\mtx
	X_0)$ as $\mtx 0\in\partial h(\mtx X),\forall \mtx X\in w(\mtx X_0)$ from \cref{lemma:2}. 
	Let $\{\mtx X_{k_j}\}$ be a convergent subsequence of $\{\mtx X_k\}$ such that $\mtx X_{k_j}\rightarrow\bar{\mtx X}$ as $j\rightarrow+\infty$. 
	Since $\mtx X_k\in\PSD,\forall k$ the decreasing property
	\eqref{le:suff} yields $\lim\limits_{k\rightarrow+\infty} h(\mtx X_k)=\lim\limits_{j\rightarrow+\infty}h(\mtx X_{k_j})=\bar{h}$. Moreover, we assume that $h(\mtx X_k)>\bar{h}$ for all $k$. Otherwise, if there exists some $k_0$ such that $h(\mtx X_{k_0})=\bar{h}$, from the decreasing property \eqref{le:suff} and $h(\mtx X_k)\geq \bar{h}$, we know $\mtx X_k=\mtx X_{k_0}$ for all $k\geq k_0$. By the definition of $w(\mtx X_0)$, we have $\lim\limits_{k\rightarrow+\infty}\dist(\mtx X_k,w(\mtx X_0))=0$. Applying the uniformized KL property  (\cite{bolte2014proximal}, Lemma 6) of $h$ on $w(\mtx X_0)$, there exist $k_\ell>0$, $\eta>0$ and $\phi\in\Phi_\eta$ such that for all $\bar{\mtx X}\in w(\mtx X_0)$, we have
	\begin{equation}\label{eq:UKL}
	\phi^{'}(h(\mtx X_k)-\bar{h})\dist(\mtx 0,\partial h(\mtx X_k))\geq 1,\ \forall k>k_\ell.
	\end{equation}
	By the inequality \eqref{le:diff}, \eqref{eq:UKL} implies
	\begin{equation}\label{eq:phi'}
	\phi^{'}(h(\mtx X_k)-\bar{h})\geq \frac{1}{b(\norm{\mtx X_{k}-\mtx X_{k-1}}_\frobenius+\bar w\norm{\mtx X_{k-1}-\mtx X_{k-2}}_\frobenius)},\ \forall k>k_\ell.
	\end{equation}
	By the concavity of $\phi$ \eqref{le:suff} and \eqref{eq:phi'}, we know that 
	\begin{equation}\label{eq:concave}
	\begin{split}
	\Delta_{k,k+1}:&=\phi(h(\mtx X_k)-\bar{h})-\phi(h(\mtx X_{k+1})-\bar{h})\\
	&\geq\phi^{'}(h(\mtx X_k)-\bar{h})(h(\mtx X_k)-h(\mtx X_{k+1}))\geq \frac{a\norm{\mtx X_{k+1}-\mtx X_k}_\frobenius^2}{b(\norm{\mtx X_{k}-\mtx X_{k-1}}_\frobenius+\bar w\norm{\mtx X_{k-1}-\mtx X_{k-2}}_\frobenius)}
	\end{split}
	\end{equation}
	Define $C=b/a>0$ in \eqref{eq:concave} and from the geometric inequality, we have
	\begin{equation}\label{eq:basicin}
	2\norm{\mtx X_k-\mtx X_{k+1}}_\frobenius\leq \norm{\mtx X_k-\mtx X_{k-1}}_\frobenius+\bar{w}\norm{\mtx X_{k-1}-\mtx X_{k-2}}_\frobenius+C\Delta_{k,k+1}
	\end{equation}
	For any $k>k_\ell$, summing up \eqref{eq:basicin} from $i=k_\ell+1,\ldots,k$, we have
	\begin{equation}\label{eq:sum}
	\begin{split}
	2\sum_{i=k_\ell+1}^k\norm{\mtx X_i-\mtx X_{i+1}}_\frobenius& \leq \sum_{i=k_\ell+1}^k \norm{\mtx X_i-\mtx X_{i-1}}_\frobenius+\bar{w}\norm{\mtx X_{i-1}-\mtx X_{i-2}}_\frobenius+C\Delta_{i,i+1} \\
	& \leq (1+\bar{w})\sum_{i=k_\ell-1}^k \norm{\mtx X_{i+1}-\mtx X_i}_\frobenius+C\Delta_{k_\ell+1,k+1}
	\end{split}
	\end{equation}
	where the last inequality is from the fact that $\Delta_{p,q}+\Delta_{q,r}=\Delta_{p,r}$ for all $p,q,r\in\mathbb{N}$. Since $\phi\geq0$ and $\bar{w}\in[0,1)$, the inequality \eqref{eq:sum} implies
	\begin{equation}\label{eq:sum1}
	(1-\bar{w})\sum_{i=k_\ell+1}^k\norm{\mtx X_i-\mtx X_{i+1}}_\frobenius\leq (1+\bar{w})\sum_{i=k_\ell-1}^{k_\ell}\norm{\mtx X_{i+1}-\mtx X_i}_\frobenius+C\phi\parens*{h\parens*{\mtx X_{k_\ell}}-\bar{h}}
	\end{equation}
	Let $k\rightarrow+\infty$ in \eqref{eq:sum1}, and thus
        $\sum_{k=1}^{+\infty}\norm{\mtx X_k-\mtx
          X_{k-1}}_\frobenius<+\infty$, which implies that $\{\mtx
        X_k\}$ converges to an $\bar{\mtx X}$, the sole member of $w(\mtx
        X_0)$.  Hence, $\vec 0\in\partial h(\bar{\mtx X})$.
\end{proof}

\section{Proof of \cref{prop:bound}}
\label{p2}
By the decreasing property \eqref{le:suff}, it is sufficient to show that the sub-level set $\brackets*{h(\mtx X)\leq h_0}$ is bounded where $h_0=h(\mtx X_0)$.
Given any $\mtx X\in\brackets*{h(\mtx X)\leq h_0}$, we have $\mtx X\in\PSD$ and $f(\mtx X)\leq h_0$. Then, 
there exists some $\mtx Y\in\mathbb{C}^{N\times N}$ such that $\mtx X = \mtx Y\mtxh Y$ with
\begin{equation}
\label{in:bound1}
\tfrac12\norm*{\mathscr{A}\parens*{\mtx Y\mtxh Y}-\vec b}_2^2 + \mu\innerproduct{\mtx R}{\mtx Y\mtxh Y\,}\leq h_0\text{.}
\end{equation}
By the triangle inequality,
\eqref{in:bound1} implies
\begin{equation*}
	\norm*{\mathscr{A}\parens*{\mtx Y\mtxh Y}}_2\leq \sqrt{2 h_0}+\norm{\vec b}_2\text{.}
\end{equation*}
as $\mtx R\in\PSD$ and $\mu\geq0$.
Standard norm inequalities gives:
\begin{equation*}
	\norm*{\mathscr{A}\parens*{\mtx Y\mtxh Y}}_1\leq \sqrt{M}\parens*{\sqrt{2h_0}+\norm{\vec b}_2}\text{.}
\end{equation*}
We note here that we can also write
\begin{equation*}
	\mathscr{A}\parens*{\mtx Y\mtxh Y} = \mtx C\brackets[\Big]{\parens*{\mtx A\mtx
			Y}\hadamardprod\parens*{\mtx A\mtx Y}^\conj}\mtx 1
\end{equation*} 
where $\hadamardprod$ denotes the element-wise (Hadamard) product,
$\mtx 1\in\realnumbers^N$ of all $1$s, and $\mtx
C\in\realmatrices{M}{MN}$ is a blockwise diagonal matrix whose $m$th
block is equal to $\mtxt 1/\sigma_m$.  Let $\sigma_{\text{max}}$
denote the minimal value among the $\sigma_m$s.  Since the bracketed
expression is the element-wise magnitude squared of $\mtx A\mtx Y$ and
hence nonnegative, we then have
\begin{equation*}
	\norm*{\mtx A\mtx Y}_\frobenius^2\leq \sigma_{\text{max}}\sqrt{M}\parens*{\sqrt{2h_0}+\norm{\vec b}_2}\text{.}
\end{equation*}
Since $\mtx A$ has full column rank, we obtain that
\begin{equation*}
	\norm*{\mtx Y}_\frobenius^2\leq \frac{\sigma_{\text{max}}\sqrt{M}}{a_{\text{min}}^2}\parens*{\sqrt{2h_0}+\norm{\vec b}_2}\text{,}
\end{equation*}
where $a_{\text{min}}>0$ the smallest nonzero singular value of $\mtx A$.
Since $\mtx X=\mtx Y\mtxh Y\,$, the above inequality implies that the
nuclear norm of $\mtx X$ is bounded:
\begin{equation*}
	\norm*{\mtx X}_*\leq \frac{\sigma_{\text{max}}\sqrt{M}}{a_{\text{min}}^2}\parens*{\sqrt{2h_0}+\norm{\vec b}_2}
\end{equation*}
and hence $\mtx X$ must be bounded assuming bounded $\vec b$ and
finite $\sigma_{\text{max}}$, and hence the sub-level set must also be bounded.

\section{Proof of \cref{prop:sub}}
\label{p1}
The dual problem of \eqref{eq:regularized} is 
\begin{equation}\label{min:dual}
\min_{\vec w,\mtx S}~\tfrac{1}{2}\norm{\vec w-\vec b}_2^2 + \iota_{\PSD}(\mtx S),\ \mbox{s.t.}\  \mathscr{A}^\HT(\vec w)+\mtx S = \mtx R.
\end{equation}
The Lagrangian function $\ell$ associated with the problem \eqref{min:dual} is 
\begin{equation*}
	\ell(\vec w, \mtx S;\mtx X)=\tfrac{1}{2}\norm{\vec w-\vec b}_2^2+\iota_{\PSD}(\mtx S) + \langle \mtx X,\mathscr{A}^\HT(\vec w)+\mtx S - \mtx R\rangle.
\end{equation*}
Since \eqref{min:dual} is strongly convex with respect to $\vec w$ and $\mtx S$ is uniquely determined by $\vec w$, thus \eqref{min:dual} admits a unique solution, denoted by $(\bar{\vec w},\bar{\mtx S})$. By the Slater's condition and $\mtx X_\star\in\mathcal{X}_\star$, the point $(\bar{\vec w},\bar{\mtx S},\mtx X_\star)$ satisfies the KKT equations:
\begin{equation}\label{eq:KKT}
0=\bar{\vec w}-\vec b+ \mathscr{A}(\mtx X_\star) ,\ 0\in\mtx X_\star+\mathcal{N}_{\PSD}(\bar{\mtx S}),\ 0=\mathscr{A}^\HT(\bar{\vec w})+\bar{\mtx S} -\mtx R.
\end{equation}
Combining the first and the last equalities in \eqref{eq:KKT},
we know $\bar{\mtx S} = (\mathscr{A}^\HT(\mathscr{A}(\mtx
X_\star)-\vec b))+\mtx R=\Grad f(\mtx X_\star)$. It is from \eqref{sub:ass} and the Proposition 3.2 in \cite{cui2016asymptotic} that $\rank(\mtx X_\star)+\rank(\bar{\mtx S})=N$. Then, applying the Corollary~3.1 in \cite{cui2016asymptotic}, we obtain that for any $\bar{\mtx X}\in\mathcal{X}_\star$, $\partial h$  is metrically subregular at $\bar{\mtx X}$ for $0$.

\section{Proof of \cref{thm:linear}}
\label{proof2}
We first present a lemma that
establishes the relationship between metric sub-regularity of
$\partial h$ and the KL inequality at critical points.
\begin{lemma}\label{lemma:subTOKL}
	Let $h=f+g$ be defined as they are in \eqref{eq:fandg} and assume $\mathcal{X}_\star\neq\emptyset$. If $\partial h$ is metrically subregular at $\bar{\mtx X}$ for $\vec 0$, then $h$ satisfies the KL inequality at $\bar{\mtx X}$ with $\phi(x)=c\sqrt{x}$ for some $c>0$.
\end{lemma}
\begin{proof} Since $h$ is convex, we know $\mathcal{X}_\star=\partial h^{-1}(\mtx 0)$. If $\partial h$ is metrically subregular at $\bar{\mtx
		X}$ for $\mtx 0$, then there exist $\kappa$ and $\epsilon>0$ such that 
	\begin{equation}\label{eq:subregular1}
	\dist(\mtx X,\mathcal{X}_\star)=\dist\parens*{\mtx X,\partial
		h^{-1}(\mtx 0)}\leq \kappa \dist(\mtx 0, \partial h(\mtx X)),\ \forall \mtx x\in\mathbb{B}\parens*{\bar{\mtx X},\epsilon}.
	\end{equation}
	Thus, for any $\mtx X\in \mathbb{B}(\bar{\mtx X},\epsilon)\cap[h(\mtx X)>h(\bar{\mtx X})]$, we have
	\begin{equation}\label{eq:subToKL}
	\begin{split}
	h(\mtx X)-h(\bar{\mtx X})& = h(\mtx X)-h(\mtx X_\star)\leq \norm{\mtx U}_\frobenius\norm{\mtx X_\star-\mtx X}_\frobenius,\ \forall \mtx X_\star\in\mathcal{X}_\star,\ \forall \mtx U\in\partial h(\mtx X),
	\end{split}	
	\end{equation}
	where the inequality is a consequence of the convexity of $h$
	and the Cauchy-Schwartz inequality. Taking the infimum over
	all $\mtx X_\star\in\mathcal{X}_\star$ and over all $\mtx
	U\in\partial h(\mtx X)$ in \eqref{eq:subToKL} and then using \eqref{eq:subregular1}, we have
	\begin{equation*}
		h(\mtx X)-h(\bar{\mtx X})\leq \kappa\dist\parens*{\mtx 0,\partial h(\mtx X)}^2,
	\end{equation*}
	which implies $h$ satisfies the KL property with $\phi(x)=2\sqrt{\kappa x}$.
\end{proof}
Now, inspired by the analysis in \cite{attouch2010proximal}, we are ready to present the proof for \cref{thm:linear}.
\begin{proof}
	Let $\bar{\mtx X}\in\mathcal{X}_\star$ such that
	$\lim\limits_{k\rightarrow+\infty}\mtx X_k=\bar{\mtx
		X}$. Assume that $h(\mtx X_k)>h(\bar{\mtx X})$ for all
	$k$. Otherwise, by the decrease property \eqref{le:suff}, it
	is easy to know $h(\mtx X_k)=h(\bar{\mtx X})$ and $\mtx
	X_k=\bar{\mtx X}$ for all $k>k_0$ whenever $h(\mtx
	X_{k_0})=h(\bar{\mtx X})$.  Define $r_k=h(\mtx
	X_k)-h(\bar{\mtx X})$. The fact that $\mtx
	X_k\rightarrow\bar{\mtx X}$ as $k\rightarrow+\infty$ combined
	with \cref{prop:sub} and \cref{lemma:subTOKL} establishes the
	existence of some $k_\ell>0$ such that the following
	inequality holds:
	\begin{equation}\label{eq:KLconv}
	h(\mtx X_k)-h(\bar{\mtx X})\leq \kappa\dist(\mtx 0,\partial h(\mtx X_k))^2, \ \forall k>k_\ell.
	\end{equation}
	Applying \eqref{le:diff} and \eqref{le:suff} to \eqref{eq:KLconv}, we obtain that
	\begin{equation}\label{ieq:r}
	\begin{aligned}
	r_k &\leq \kappa b^2\parens*{\norm{\mtx X_k-\mtx X_{k-1}}_\frobenius+\bar{w}\norm{\mtx X_{k-1}-\mtx X_{k-2}}_\frobenius}^2 \\
	& \leq 2\kappa b^2\parens*{\norm{\mtx X_k-\mtx X_{k-1}}_\frobenius^2+\bar{w}^2\norm{\mtx X_{k-1}-\mtx X_{k-1}}_\frobenius^2}\\
	& \leq 2\kappa a^{-1}b^2\braces*{F(\mtx X_{k-1})-F(\mtx
		X_{k})+\bar{w}^2\brackets[\Big]{F(\mtx X_{k-2})-F(\mtx X_{k-1})}}
	\\
	&= c\parens[\big]{r_{k-1}-r_k+\bar{w}(r_{k-2}-r_{k-1})}
	\end{aligned}
	\end{equation}
	where $c=2\kappa a^{-1}b^2$ and the second inequality is from the geometric inequality. Since $r_k\leq r_{k-1}$ for all $k$, the inequality \eqref{ieq:r} implies
	\begin{equation*}
		r_k\leq \frac{c}{1+c}\brackets[\bigg]{(1-\bar{w})r_{k-1}+\bar{w}r_{k-2}}\leq\frac{cr_{k-2}}{1+c}\leq r_{k_\ell}\parens[\bigg]{\frac{c}{1+c}}^{\frac{k-k_\ell-1}{2}},\ \forall k>k_{\ell}.
	\end{equation*}
	Furthermore, using \eqref{eq:sum1}, for all $\hat k>k>k_\ell$, we have
	\begin{equation}\label{conv:1}
	\begin{aligned}
	(1-\bar{w})\norm{\mtx X_{\hat k}-\mtx X_{k}}_\frobenius&
	\leq(1-\bar{w})\sum_{i=k}^{\hat k-1}\nolimits\norm{\mtx X_{i+1}-\mtx X_i}_\frobenius\\
	&\leq (1+\bar{w})\sum_{i=k-2}^{k}\nolimits\norm{\mtx X_{i+1}-\mtx X_i}_\frobenius + b\sqrt{r_{k}}/a\\
	&\leq (1+\bar{w})\sum_{i=k-2}^{k}\sqrt{\frac{h(\mtx X_{i+1})-h(\mtx X_i)}{a}}+ b\sqrt{r_k}/a\leq \tilde{\nu}\sqrt{r_{k-2}}
	\end{aligned}
	\end{equation}
	where $\tilde{\nu}=(1+\bar{w})\sqrt{2/a}+b/a$. Letting $\hat
	k\rightarrow+\infty$ in \eqref{conv:1}, we thus know that for all $k>k_\ell+2$
	\begin{equation*}
	\norm{\mtx X_k-\bar{\mtx X}}_\frobenius\leq \nu\sqrt{r_{k_\ell}}\bigg(\frac{c}{1+c}\bigg)^{\frac{k-k_\ell-3}{4}},
	\end{equation*}
	where $\nu = (1-\bar{w})^{-1}\tilde{\nu}$.
\end{proof}

\end{document}

%% file: ex_shared.tex
\usepackage{lipsum}
\usepackage{amsmath}
\usepackage{amsfonts}
\usepackage{notation}
\usepackage{graphicx}
\usepackage{epstopdf}
\usepackage{algorithm}
\usepackage{algpseudocode}

\ifpdf
  \DeclareGraphicsExtensions{.eps,.pdf,.png,.jpg}
\else
  \DeclareGraphicsExtensions{.eps}
\fi

\numberwithin{theorem}{section}
\renewtheorem{lemma}{Lemma}[section]
\newtheorem{remark}{Remark}[section]

\newtheorem{assume}{Assumption}[section]
\renewtheorem{proposition}{Proposition}[section]
\renewtheorem{corollary}{Corollary}[section]

\DeclareMathOperator{\Prox}{Prox}
\DeclareMathOperator{\Proj}{Proj}
\DeclareMathOperator{\rank}{rank}
\DeclareMathOperator{\dist}{dist}
\DeclareMathOperator{\ri}{ri}
\DeclareMathOperator{\gph}{gph}
\DeclareMathOperator{\dom}{dom}
\DeclareMathOperator{\aff}{aff}
\newcommand{\SN}{\mathcal{S}^N}
\newcommand{\PSD}{\mathcal{S}_{+}^N}
\newcommand{\A}[1]{\mathscr{A}\parens{#1}}

\newcommand{\TheTitle}{Coherence retrieval using trace regularization}
\newcommand{\TheAuthors}{C. Bao, G. Barbastathis, H. Ji, Z. Shen, Z. Zhang}

\headers{\TheTitle}{\TheAuthors}

\title{{\TheTitle}\thanks{Submitted to the editors April 3, 2017.
\funding{The work of the authors was partially supported by Singapore
MOE Research Grant MOE2014-T2-1-065 and MOE2012-T3-1-008.  The
research was supported by the National Research Foundation (NRF),
Prime Minister’s Office, Singapore, under its CREATE programme,
Singapore-MIT Alliance for Research and Technology (SMART) BioSystems
and Micromechanics (BioSyM) IRG.}}}

\author{
  Chenglong Bao\thanks{Department of Mathematics, National University of Singapore
    (\email{matbc, matjh, matzuows@nus.edu.sg}).}
  \and
  George Barbastathis\thanks{Singapore-MIT Alliance for Research and Technology (SMART) Centre, and Department of Mechanical Engineering, Massachusetts Institute of Technology,  (\email{gbarb@mit.edu}).}
  \and
  Hui Ji\footnotemark[2]
  \and
  Zuowei Shen\footnotemark[2]
  \and
  Zhengyun Zhang\thanks{Singapore-MIT Alliance for Research and Technology (SMART) Centre, (\email{zhengyun@smart.mit.edu}).}
}

\usepackage{amsopn}